\documentclass[11pt]{article}
\pdfoutput=1
\usepackage[utf8]{inputenc}
\usepackage[text={6in,9.5in}]{geometry}

\usepackage{amsmath, amsfonts,amssymb, amsthm}
\usepackage{thm-restate}

\usepackage{enumitem}
\setlist{topsep=3pt,itemsep=0pt}

\usepackage[english]{babel}
\usepackage{xcolor}

\usepackage{algorithm}
\usepackage{algorithmicx}
\usepackage{algpseudocode}
\algtext*{EndWhile}
\algtext*{EndIf}
\algtext*{EndFor}
\algtext*{EndFunction}
\algtext*{EndProcedure}

\newtheorem{theorem}{Theorem}
\newtheorem{definition}[theorem]{Definition}
\newtheorem{lemma}[theorem]{Lemma}

\newtheorem{obs}[theorem]{Observation}

\newcommand{\ceil}[1]{{\left\lceil#1\right\rceil}}

\newcommand{\set}[1]{{\left\{#1\right\}}}

\newcommand{\swap}{\mathsf{swap}}
\newcommand{\open}{\mathsf{open}}
\newcommand{\close}{\mathsf{close}}

\newcommand{\R}{{\mathbb{R}}}

\newcommand{\cc}{\mathsf{cc}}

\mathchardef\mhyphen="2D

\newcommand{\cost}{\mathsf{cost}}
\newcommand{\poly}{\mathrm{poly}}

\newcommand{\last}{{\mathsf{last}}}

\newcommand{\FL}{{\mathsf{FL}}}

\newcommand{\calS}{\mathcal{S}}
\newcommand{\calP}{{\mathcal{P}}}
\newcommand{\calQ}{{\mathcal{Q}}}

\newcommand{\opt}{\mathsf{opt}}

\newcommand{\LB}{{\mathsf{LB}}}

\newcommand{\level}{{\mathsf{level}}}

\newcommand{\rmF}{{\mathrm{F}}}
\newcommand{\rmC}{{\mathrm{C}}}
\newcommand{\best}{{\mathrm{best}}}

\DeclareMathOperator*\E{\mathbb{E}}

\ifdefined\DEBUG
	\newcommand{\sh}[1]{{\color{red} #1}}

	\def\rem#1{{\marginpar{\raggedright\scriptsize #1}}}
	\newcommand{\shr}[1]{\rem{\small\textcolor{red}{$\bullet${\tiny Shi: #1\\}}}}
	\newcommand{\janr}[1]{\rem{\small\textcolor{brown}{$\bullet${\tiny Janard: #1\\}}}}
	\newcommand{\xguor}[1]{\rem{\small\textcolor{blue}{$\bullet${\tiny Xiangyu: #1\\}}}}

	\newcommand{\remove}[1]{{\color{lightgray} #1}}
	
\else
	\newcommand{\sh}[1]{#1}

	\newcommand{\shr}[1]{}
	\newcommand{\janr}[1]{}
	\newcommand{\xguor}[1]{}
	\newcommand{\jaiyir}[1]{}
	
	\newcommand{\remove}[1]{}
	
\fi

\begin{document}

\title{The Power of Recourse: Better Algorithms for Facility Location in Online and Dynamic Models}

\author{ 
		Xiangyu Guo \thanks{Department of Computer Science and Engineering, University at Buffalo, {\tt xiangyug@buffalo.edu }} \and
		Janardhan Kulkarni \thanks{The Algorithms Group, Microsoft Research, Redmond, {\tt jakul@microsoft.com}} \and
		Shi Li \thanks{Department of Computer Science and Engineering, University at Buffalo, {\tt shil@buffalo.edu}} \and
		Jiayi Xian \thanks{Department of Computer Science and Engineering, University at Buffalo, {\tt jxian@buffalo.edu}}
	   }

\date{}

\maketitle
\begin{abstract}
	In this paper we study the facility location problem in the online with recourse and dynamic algorithm models. In the online with recourse model, clients arrive one by one and our algorithm needs to maintain good solutions at all time steps with only a few changes to the previously made decisions (called recourse).  
	We show that the classic local search technique can lead to a $(1+\sqrt{2}+\epsilon)$-competitive online algorithm for facility location with only $O\left(\frac{\log n}{\epsilon}\log\frac1\epsilon\right)$ amortized facility and client recourse.
	
	We then turn to the dynamic algorithm model for the problem, where the main goal is to design fast algorithms that maintain good solutions at all time steps.  We show that the result for online facility location, combined with the randomized local search technique of Charikar and Guha \cite{CharikarGhua2005}, leads to an $O(1+\sqrt{2}+\epsilon)$ approximation dynamic algorithm with amortized update time of $\tilde O(n)$ in the incremental setting against adaptive adversaries. Notice that the running time is almost optimal, since in general metric space, it takes $\Omega(n)$ time to specify a new client's position.
	The approximation factor of our algorithm also matches the best offline analysis of the classic local search algorithm.
	
	Finally, we study the fully dynamic model for facility location, where clients can both arrive and depart. Let $F$ denote the set of available facility locations. Our main result is an $O(1)$-approximation algorithm in this model with $O(|F|)$ preprocessing time and $O(\log^3 D)$ amortized update time for the HST metric spaces. Using the seminal results of Bartal \cite{Bartal96} and Fakcharoenphol, Rao and Talwar \cite{Fakcharoenphol2003}, which show that any arbitrary $N$-point metric space can be embedded into a distribution over HSTs such that the expected distortion  is  at most $O(\log N)$, we obtain a $O(\log |F|)$ approximation with preprocessing time of $O(|F|^2\log |F|)$ and $O(\log^3 D)$ amortized update time. The approximation guarantee holds in expectation for every time step of the algorithm, and the result holds in the oblivious adversary model. 
\end{abstract} 
\section{Introduction}
\label{sec:intro}
In the {\em (uncapacitated) facility location problem}, we are given a metric space $(F \cup C,d)$, where $F$ is the set of facility locations, $C$ is the set of clients, and $d: (F \cup C) \times (F \cup C) \rightarrow \R_{\geq 0}$ is a distance function, which is non-negative, symmetric and satisfies triangle inequalities.  For each location $i \in F$, there is a facility opening cost $f_i \geq 0$. 
The goal is open a subset $S \subseteq F$ of facilities so as to minimize cost of opening the facilities and the connection cost. 
The cost of connecting a client $j$ to an open facility $i$ is equal to $d(j,i)$. 
Hence, the objective function can be expressed concisely as $\min_{S\subseteq F} \left(f(S)  + \sum_{j \in C}d(j, S)\right)$, where for a set $S \subseteq F$, $f(S) := \sum_{i \in S}f_i$ is the total facility cost of $S$ and $d(j, S):=\min_{i \in S}d(j, i)$ denotes the distance of $j$ to the nearest location in $S$. 
The facility location problem arises in countless applications: in the placement of servers in data centers, network design, wireless networking, data clustering, location analysis for placement of fire stations, medical centers, and so on.
Hence, the problem has been studied extensively in many different communities: approximation algorithms, operations research, and computational geometry.
In the approximation algorithms literature  in particular, the problem occupies a prominent position as the  development of every major technique in the field is tied to its application on the facility location problem. See the text book by Williamson and Shmoys \cite{Williamson} for more details.
The problem is hard to approximate to a factor better than 1.463 \cite{Guha1998}. 
The current best-known polynomial-time algorithm is given by the third author, and achieves 1.488-approximation \cite{Li13}.

In many real-world applications the set of clients arrive online, the metric space can change over time, and there can be memory constraints:
This has motivated the problem to be studied in various models: online \cite{Meyerson2001,Fotakis08algorithmica,Anagnostopoulos2004,Fotakis2007}, dynamic \cite{Cohen-Addad19, Goranci19,CyganCMS18,Wesolowsky:1973,Farahani2009, Eisenstat,AnNS15}, incremental \cite{Fotakis06,Charikar1997,Fotakis2011}, streaming \cite{Indyk2004, Fotakis11, Lammersen, Czumaj2013, Charikar1997}, game theoretic \cite{Vetta2002,FotakisT13,FotakisT13a}, to name a few. 
This paper is concerned with {\em online and dynamic models}. 
Thus to keep the flow of presentation linear, we restrict ourselves to the results in these two models here.
\remove{In the Section , we put our results in the broader context. }

\medskip
Motivated by its applications in network design and data clustering, Meyerson \cite{Meyerson2001} initiated the study of facility location problem in the online setting.
Here, clients arrive online one-by-one, the algorithm has to assign the newly arriving client to an already opened facility or needs to open a new facility to serve the request.
The decisions made by the algorithm are {\em irrevocable}, in the sense that a facility that is opened cannot be closed and the clients cannot be reassigned.
In the online setting, Meyerson \cite{Meyerson2001} designed a very elegant randomized algorithm that achieves an $O(\log n)$ competitive ratio, and  also showed that no online algorithm can obtain $O(1)$ competitive ratio.
This result was later extended by Fotakis \cite{Fotakis08algorithmica} to obtain an {\em asymptotically optimal} $O(\log n/\log \log n)$-competitive algorithm.
Both the algorithms and analysis techniques in \cite{Fotakis08algorithmica, Meyerson2001} were influential, 
and found many applications in other models such as streaming \cite{Fotakis2011}.
\shr{I think the lower bound was shown in a different paper; and the bound of Meyerson was $O(\log n)$.} The lowerbound in Fotakis \cite{Fotakis08algorithmica}  holds even in very special metric spaces such as HSTs or the real line. 
Since then, several online algorithms have been designed achieving the same competitive ratio with more desirable properties such as deterministic \cite{Anagnostopoulos2004}, primal-dual \cite{Fotakis2007}, or having a small memory footprint \cite{Fotakis11}.
We refer to a beautifully written survey by Fotakis \cite{Fotakis2011} for more details.

The main reason to assume that decisions made by an algorithm are irrevocable is because the cost of changing the solution is expensive in some applications.
However, if one examines these above applications closely, say for example connecting clients to servers in data centers, it is more natural to assume that decisions need not be irrevocable but the algorithm {\em should not change the solution too much}.
This is even more true in modern data centers where topologies can be reconfigured; see \cite{GhobadiMPDKRBRG16} for more details.
A standard way of quantifying the restriction that an online algorithm does not make too many changes is using the notion of {\em recourse}. 
The recourse per step of an online algorithm is the {\em number of changes} it makes to the solution.
Recourse captures the {\em minimal} amount of changes an online algorithm {\em has to make} to maintain a  desired competitive ratio due to the {\em information theoretic} limits.
For the facility location problem, depending on the application, the recourse can correspond to: 1) the number of changes made to the opened facilities (called \emph{facility recourse})  2) the number of reconnections made to the clients (called \emph{client recourse}).
\sh{Notice that we can assume for every facility we open/close, we have to connect/disconnect at least one client. Thus the client recourse is at least the facility recourse.}
In the clustering applications arising in massive data sets, the opened facilities represent cluster centers, which represent summaries of data. Here one is interested in making sure that summaries do not change too frequently as more documents are added online. 
Therefore, facility recourse is a good approximation to the actual cost of changing the solution \cite{Charikar1997,Fotakis06}. 
On the other hand, in network design problems, client recourse is the true indicator of the cost to implement the changes in the solution.
As a concrete example, consider the problem of connecting clients to servers in datacenters, which was one of the main motivation for Meyerson \cite{Meyerson2001} to initiate the study of online facility location problem.
Here, it is important that one does not reconnect clients to servers too many times, as such changes can incur significant costs both in terms of disruption of service and the labor cost. \sh{Consider another scenario where a retailing company tries to maintain stores to serve the dynamically changing set of clients.  As the clients are changing so frequently, it would be infeasible to build/shutdown even one store for every new client. In this application, small client recourse per step is desirable, as that will automatically forbid frequent changes of store locations.}\xguor{Discuss the NeurIPS'19 paper on fully dynamic facility location by Cohen-Addad et~al.\cite{Cohen-Addad19}}

In this light, a natural question that arises is:   \shr{Also the facility recourse is always smaller than the client recourse since every time we open or close a facility, at least one client will be reconnected. Do we need to mention facility recourse in the results? Can we say this here? }

\vspace{2mm}
{\em Is it possible to maintain a constant approximation for the facility location problem if we require that  the facility and client recourse is small?}   
\vspace{2mm}

Our first main result shows that indeed this is possible. In the following theorems, we use $n$ to denote the total number of facility locations and all clients that ever arrived, and $D$ to denote the diameter of the metric $d$ (assuming all distances are integers).
\begin{theorem}
\label{UFL-recourse}
There is a deterministic online algorithm for the facility location problem that achieves a competitive ratio of $(1+\sqrt{2} + \epsilon)$ with $O\left(\frac{\log n}{\epsilon}\log\frac1\epsilon\right)$ amortized facility and client recourse against an adaptive adversary.
\end{theorem}

Our algorithm to show the above theorem differs from the previous approaches used in the context of online variants of facility location problem, and is based on {\em local search}. The local search algorithm is one of the most widely used algorithms for the facility location problem in practice and is known to achieve an approximation factor of $(1+\sqrt 2)$ in the offline setting. See the influential paper by Arya {\em et al} \cite{AryaGKMP01} and a survey by Munagala \cite{Munagala16}. Thus our result matches the best known approximation ratio  for offline facility location using local search. 
Further, our result shows that the local search algorithm augmented with some small modifications is inherently {\em stable}  as it does not make too many changes to the solutions even if clients are added in an online fashion. This gives further justification for its popularity among practitioners.

Prior to Theorem \ref{UFL-recourse}, the known results \cite{Fotakis06, Diveki2011,Fotakis11} needed one or more of these assumptions: 1)  the facility costs are {\em the same} 2) we are interested in knowing only the cost of solution 3) we are interested only in bounding the {\em facility recourse}. 
In particular, there was no known algorithm that bounds the client recourse, which is an important consideration in many applications  mentioned above.
Moreover, our algorithm also achieves a better approximation factor;
previously best known algorithm for the facility location problem achieved a competitive ratio of 48 \cite{Fotakis2011}.

Our result in the recourse setting for the facility location problem should be contrasted with the similar results shown recently for online Steiner tree \cite{Gupta015}, set cover \cite{GuptaK0P17}, scheduling \cite{GuptaKS14}, and matchings and flows \cite{BernsteinHR19,GuptaKS14}. 
Moreover, these results also raise an intriguing questions: {\em is polylog amount of recourse enough to beat information theoretic lowerbounds in the online algorithms? Is recourse as or more powerful than randomization?}

\medskip

While having a small client recourse is enough in data center applications, it is not enough in some others.
Take wireless networks as a concrete example. 
Here, the set of clients (mobile devices) keeps changing over time, and it is necessary to {\em update} the assignment of clients to facilities as {\em  quickly} as possible so to minimize the service disruption.
These applications motivated Cygan {\em et~al}~\cite{CyganCMS18}, Goranci {\em et~al}~\cite{Goranci19} and Cohen-Addad {\em et~al}~\cite{Cohen-Addad19} to study the facility location problem in the framework of {\em dynamic algorithms}. 
The dynamic model of \cite{CyganCMS18} and \cite{Cohen-Addad19} is different from what we study here, so we discuss it at end of this section.

The dynamic facility location problem is similar to the one in online setting except that at each time step either {\em a new client arrives or an existing client departs}.
The goal is to always maintain a solution that is a constant factor approximation to the optimal solution, while minimizing {\em the total time spent in updating the solution.}
We emphasize that we require our dynamic algorithms to maintain {\em an actual assignment of clients to facilities}, not just the set of open facilities and an estimate of connection cost.
This is important for applications mentioned above.
This setting was considered in \cite{Goranci19}, who showed that for metric spaces with {\em doubling dimension $\kappa$}, there is a deterministic fully dynamic algorithm with $\tilde O(2^{\kappa^2})$ update time, which maintains a constant approximation. 
However, for more general metric spaces no results were known in the dynamic setting, and we give the first results.	
First we consider the incremental setting, where clients only arrive and never depart.

\begin{theorem}
\label{UFL-dynamicIncremental}
In the incremental setting against an adaptive adversary, there is a randomized dynamic algorithm for the facility location problem that, with probability at least $1-1/n^2$, maintains an approximation factor of $(1+\sqrt{2} + \epsilon)$ and has \emph{total} update time of $O(\frac{n^2}{\epsilon^2}\log^3n\log\frac1\epsilon)$. 
\end{theorem}

Note that it takes $\Theta(n|F|)$ space to specify the input in our model (see Section~\ref{subsec:specify-input}). Hence the  running time of our algorithms is almost optimal up to polylog factors when $|F| = \Omega(n)$.   The proof of above theorem uses randomized local search and builds on our result in the recourse setting. 
We use randomization to convert the recourse bound into an update time bound. 
Further, our analysis of above theorem also implies one can obtain $O(\frac{n|F|}{\epsilon^2}\log^3n\log\frac1{\epsilon})$ running time by losing $O(1)$ factors in the approximation ratio; 
see the remark at the end of Section \ref{sec:dfl}.

\medskip
Next we study the fully dynamic setting. 
Here, we first consider an important class of metric spaces called hierarchically well separated tree (HST) metrics {\cite{Bartal96}; see Definition~\ref{def:HST} for the formal definition, and Section~\ref{subsec:specify-input} for more details about how the input sequence is given.
\remove{A $k$-hierarchically well-separated tree is defined as a rooted
weighted tree with following properties: 1) The edge weight from any node to each of its children is same. 2)
The edge weights along any path from the root to a leaf are decreasing by a factor
of at least $k$.
In most applications, $k$ is assumed to be a small constant.} \shr{This is different from what I defined in the preliminary section.}
For HST metric spaces, we show the following result. \remove{In the following result, $n$ denotes the total number of clients that arrive in the entire course of the algorithm. }\shr{This was already defined.}


\begin{theorem}
	\label{UFL-HST}
	In the fully dynamic setting against adaptive adversaries, there is a deterministic algorithm for the facility location problem that achieves an $O(1)$ approximation factor with $O(|F|)$ preprocessing time and $O(n\log^3 D)$ total update time for the HST metric spaces.
\end{theorem}

A seminal result by Bartal \cite{Bartal96}, which was later tightened by Fakcharoenphol, Rao and Talwar \cite{Fakcharoenphol2003}, shows that any arbitrary $N$-point metric space can be embedded into a distribution over HSTs such that the expected distortion  is  at most $O(\log N)$, which is also tight.
Moreover, such a probabilistic embedding can also be computed in $O(N^2\log N)$ time; see recent results by  Blelloch, Gu and Sun for details \cite{Blelloch0S17}.
These results immediately imply the following theorem, provided the input is specified as in Section~\ref{subsec:specify-input}.

\begin{theorem}
\label{UFL-fullydynamic}
In the fully dynamic setting against oblivious adversary, there is a randomized algorithm for the facility location problem that maintains an approximation factor of $O(\log |F|)$ with \sh{preprocessing time of $O(|F|^2\log |F|)$} and $O(n\log^3 D)$ total update time. The approximation guarantee holds only in expectation for every time step of the algorithm.
\end{theorem}
Observe that unlike the incremental setting, the above theorem holds only in the oblivious adversary model, as probabilistic embedding techniques preserve distances only in expectation as can be seen by taking a cycle on $n$ points.
Our result also shows that probabilistic tree embeddings using HSTs can be a very useful  technique in the design of dynamic algorithms,
similar to its role in online algorithms \cite{Bartal96, BartalBBT97,  Umboh15, BubeckCLLM18}.

\medskip 

Our algorithms in Theorems \ref{UFL-HST} and \ref{UFL-fullydynamic} in the fully dynamic setting also have the nice property that amortized client and facility {\em recourse} is $O(\log^3D)$ (in fact, we can achieve a slight better bound of $O(\log^2 D)$ as can be seen from the analysis). This holds as our dynamic algorithms maintain the entire assignment of clients to facilities {\em explicitly} in memory at every time step. Thus, the amortized client reconnections is at most the amortized update time.  This is useful when one considers an online setting where clients arrive and depart, and is interested in small client recourse. 
%
%
A fully dynamic online model of facility location problem, where clients arrive and \emph{depart} was recently studied by Cygan {\em et~al}~\cite{CyganCMS18} and Cohen-Addad {\em et~al}~\cite{Cohen-Addad19}, but with different assumption on recourse. 
In this model, when a client arrives, the algorithm has to assign it to an open facility immediately; While upon departure of a client, if a facility was opened at the same location, then the clients that were assigned to that location should be reassigned immediately and irrevocably.
Cygan {\em et~al}~\cite{CyganCMS18} studied the case when recourse is \emph{not} allowed: they showed that a delicate extension of Meyerson's \cite{Meyerson2001} algorithm obtains asymptotically tight competitive ratio of $O(\log n /\log \log n)$. Cohen-Addad {\em et~al}~\cite{Cohen-Addad19} later showed that this can be improved to $O(1)$ if recourse is allowed.
However, both results holds only for the {\em uniform facility costs} and Cygan {\em et~al}\cite{CyganCMS18} even showed an {\em unbounded} lower bound for the non-uniform facility cost case in their model. Moreover, in their model reconnections of clients are assumed to be ``automatic'' and do not count towards the client recourse; it is not clear how many client reconnections their algorithm will make.

\subsection{Our Techniques}
Our main algorithmic technique for proving Theorems~\ref{UFL-recourse} and \ref{UFL-dynamicIncremental} is local search, which is one of the powerful algorithm design paradigms.
%
\sh{Indeed, for both results, the competitive (approximation) ratio we achieve is $1+\sqrt{2}+\epsilon$, which matches the best approximation ratio for offline facility location obtained using local search \cite{AryaGKMP01}.
Both of our results are based on the following key lemma. Suppose we maintain local optimum solutions at every time step in our algorithm. When a new client $j_t$ comes at time $t$, we add it to our solution using a simple operation, and let $\Delta_t$ be the increase of our cost due to the arrival of $j_t$.  The key lemma states that the sum of $\Delta_t$ values in the first $T'$ time steps can be bounded in terms the optimum cost at time $T'$. With a simple modification to the local search algorithm, in which we require each local operation decreases enough cost for every client it reconnects, one can bound the total client recourse. }

\sh{ The straightforward way to implement the local search algorithm takes time $\Omega(n^3)$. To derive a better running time, we leverage the randomized local search idea of Charikar and Guha \cite{CharikarGhua2005}. At every iteration, \sh{we randomly choose a facility $i$ or a closing operation, and then perform the best operation that opens or swaps in $i$, or closes a facility if that is what we choose}.   By restricting the facility $i$ and with the help of the heap data structure, an iteration of the algorithm can be implemented in time $O(|C|\log |F|)$.  As in \cite{CharikarGhua2005} we can also show that each iteration can make a reasonable progress in expectation, leading to a bound of $\tilde O(|F|)$ on the number of iterations for the success of the algorithm with high probability. We remark that the algorithm in \cite{CharikarGhua2005} used a different local search framework. Therefore, our result shows that the classic algorithm of \cite{AryaGKMP01} can also be made fast.  }

\sh{However, directly replacing the randomized local search procedure with a deterministic one does not work: The solution at the end of each time might not be a local optimum as we did not enumerate all possible local operations. Thus the key lemma does not hold any more. Nevertheless we show that applying a few local operations around $j_t$ upon its arrival can address the issue.  With the key lemma, one can bound the number of times we perform the iterative randomized local search procedure, and thus the overall running time. }

\sh{Our proof for Theorem~\ref{UFL-HST} is based on a generalization of the greedy algorithm for facility location on HST metrics, which was developed in \cite{EsencayiGLW19} in the context of {\em differential privacy} but only for the case of {\em uniform} facility cost. The intuition of the algorithm is as follows: If for some vertex $v$ of the HST $T$, the number of clients in the tree $T_v$ (the sub-tree of $T$ rooted at $v$) times the length of parent edge of $v$ is big compared to the cost of the cheapest facility in $T_v$, then we should open that facility. Otherwise, we should not open it and let the clients in $T_v$ be connected to outside $T_v$ through the parent edge.  This intuition can be made formal: We mark $v$ in the former case; then simply opening the cheapest facility in $T_v$ for all \emph{lowest marked} vertices $v$ leads to a constant approximation for facility location. }

\sh{The above offline algorithm  leads to a \emph{dynamic data structure} that maintains $O(1)$-approximate solutions, supports insertion and deletion of clients, and reports the connecting facility of a client in $O(\log D)$ time. This is the case since each time a client arrives or departs, only its ancestors will be affected.  However, in a dynamic algorithm setting, we need to maintain the assignment vector in memory, so that when the connecting facility of a client changes, it needs to be notified.  This requires that the number of reconnections made by our algorithm to be small.  To achieve the goal, we impose two constants for each $v$ when deciding whether $v$ should be marked and the cheapest facility in $T_v$ should be open. When a vertex $v$ changes its marking/opening status, we update the constants in such a way that it becomes hard for the status to be changed back. 
}

\section{Preliminaries}
\label{sec:prelim}

		
	Throughout the paper, we use $F$ to denote the set of potential facilities for all the problems and models; we assume $F$ is given upfront.  
	$C$ is the dynamic set of clients we need to connect by our algorithm. This is not necessarily the set of clients that are present: In the algorithms for online facility location with recourse and dynamic facility location in the incremental setting, we fix the connections of some clients as the algorithms proceed. These clients are said to be ``frozen'' and excluded from $C$. We shall always use $d$ to denote the hosting metric containing $F$ and all potential clients. For any point $j$ and subset $V$ of points in the metric, we define $d(j, V) = \min_{v \in V}d(j, v)$ to be the minimum distance from $j$ to a point in $V$. We assume all distances are integers, the minimum non-zero distance between two points is 1. We define $D$, the diameter or the aspect ratio of a metric space, as the largest distance between two points in it.  Let $n$ be $|F|$ plus the total number of clients arrived during the whole process.  The algorithms do not need to know the exact value of $n$ in advance, except that in the dynamic algorithm for facility location in the incremental setting (the problem in Theorem~\ref{UFL-dynamicIncremental}), to achieve the $1- 1/n^2$ success probability, a sufficiently large $\Gamma  = \poly(n, \log D, \frac1\epsilon)$ needs to be given.\footnote{For an algorithm that might fail, we need to have some information about $n$ to obtain a failure probability that depends on $n$.} \shr{TODO: remove $n_c$ and $n_f$. $n_c$ is never used. We can just use $|F|$.Done. }
	
	In all the algorithms, we maintain a set $S$ of open facilities, and a connection $\sigma \in S^C$ of clients in $C$ to facilities in $S$. We do not require that $\sigma$ connects clients to their respective nearest open facilities.   For any solution $(S' \subseteq F, \sigma' \in S'^C)$, we use $\cc(\sigma') = \sum_{j \in C}d(j, \sigma_j)$ to denote the connection cost of the solution. For facility location, we use $\cost(S', \sigma') = f(S') + \cc(\sigma')$ to denote the total cost of the solution $(S', \sigma')$, where $f(S') := \sum_{i \in S'} f_i$.  Notice that $\sigma$ and the definitions of $\cc$ and $\cost$ functions depend on the dynamic set $C$.

	Throughout the paper, we distinguish between a ``moment'', a ``time'' and a ``step''. A moment refers to a specific time point during the execution of our algorithm.  A time corresponds to an arrival or a departure event:  At each time, exactly one client arrives or departs, and time $t$ refers to the period from the moment the $t$-th event happens until the moment the $(t+1)$-th event happens (or the end of the algorithm). One step refers to one statement in our pseudo-codes indexed by a number. 

\subsection{Hierarchically Well Separated Trees}
\begin{definition} \label{def:HST}
	A hierarchically-well-separated tree (or HST for short) is an edge-weighted rooted tree with the following properties:
	\begin{itemize}[topsep=3pt,itemsep=0pt]
		\item all the root-to-leaf paths have the same number of edges,
		\item if we define the level of  vertex $v$, $\level(v)$, to be the number of edges in a path from $v$ to any of its leaf descendant, then for an non-root vertex $v$, the weight of the edge between $v$ and its parent is exactly $2^{\level(v)}$. 
	\end{itemize}
	Given a HST $T$ with the set of leaves being $X$, we use $d_T$ to denote the shortest path metric of the tree $T$ (with respect to the edge weights) restricted to $X$.  
\end{definition}

The classic results by  Bartal \cite{Bartal96} and Fakcharoenphol, Rao and Talwar \cite{Fakcharoenphol2003} state that we can embed any $N$-point metric $(X, d)$ (with minimum non-zero distance being $1$) to a distribution $\pi$ of \emph{expanding}\footnote{A metric $(X, d_T)$ is expanding w.r.t $(X, d)$ if for  every $u, v \in X$, we have $d_T(u, v) \geq d(u, v)$.} HST metrics $(X, d_T)$  with distortion $O(\log N)$: For every $u, v \in X$, we have $d_T(u, v) \geq d(u, v)$ and $\E_{u, v}[d_T(u, v)] \leq O(\log N) d(u, v)$. Moreover, there is an efficient randomized algorithm \cite{Blelloch0S17} that outputs a sample of the tree $T$ from $\pi$.  Thus applying standard arguments, Theorem~\ref{UFL-HST} implies Theorem~\ref{UFL-fullydynamic}.

\subsection{Specifying Input Sequence} \label{subsec:specify-input}
In this section we specify how the input sequence is given.  For the online and dynamic facility location problem, we assume the facility locations $F$, their costs $(f_i)_{i \in F}$, and the metric $d$ restricted to $F$ are given upfront, and they take $O(|F|^2)$ space.  Whenever a client $j \in C$ arrives, it specifies its distance to every facility $i \in F$ (notice that the connection cost of an assignment $\sigma \in S^C$ does not depend on distances between two clients and thus they do not need to be given). Thus the whole input contains $O(n|F|)$ words. 

For Theorems~\ref{UFL-HST} and \ref{UFL-fullydynamic}, as we do not try to optimize the constants, we {\em do not} need that a client specifies its distance to every facility.
By losing a multiplicative factor of $2$ and an additive factor of $1$ in the approximation ratio, we can assume that every client $j$ is collocated with its nearest facility in $F$ (See Appendix~\ref{appendix:moving-clients}).  Thus, we only require that when a client $j$ comes, it reports the position of its nearest facility. For Theorem~\ref{UFL-HST}, the HST $T$ over $F$ is given at the beginning using $O(|F|)$ words. For Theorem~\ref{UFL-fullydynamic}, the metric $d$ over $F$ is given at the beginning using $O(|F|^2)$ words. Then, we use an efficient algorithm \cite{Blelloch0S17} to sample a HST $T$.

\subsection{Local Search for facility location}
The local-search technique has been used to obtain the classic $(1+\sqrt 2)$-approximation offline algorithm for facility location \cite{AryaGKMP01}. We now give an overview of the algorithm, which will be the baseline of our online and dynamic algorithms for facility location.   One can obtain a (tight) $3$-approximation for facility location without scaling facility costs. Scaling the facility costs by a factor of $\lambda := \sqrt{2}$ when deciding whether an operation can decrease the cost, we can achieve a better approximation ratio of $\alpha_\FL:= 1+\sqrt{2}$.  Throughout, we fix the constants $\lambda = \sqrt{2}$ and $\alpha_\FL = 1+\sqrt{2}$.  For a solution $(S', \sigma')$ to a facility location instance, we use $\cost_\lambda(S', \sigma') =  \lambda f(S') + \cc(\sigma')$ to denote the cost of the solution $(S', \sigma')$ with facility costs scaled by $\lambda = \sqrt{2}$. We call $\cost_\lambda(S', \sigma')$ the \emph{scaled cost} of $(S', \sigma')$.

Given the current solution $(S, \sigma)$ for a facility location instance defined by $F, C, d$ and $(f_i)_{i \in F}$, we can apply a \emph{local operation} that changes the solution $(S, \sigma)$.  A valid local operation is one of the following.
\begin{itemize}
	\item An $\open$ operation, in which we open some facility $i \in F$ and reconnect a subset $C' \subseteq C$ of clients to $i$.  We allow $i$ to be already in $S$, in which case we simply reconnect $C'$ to $i$. This needs to be allowed since our $\sigma$ does not connect clients to their nearest open facilities.
	\item A $\close$ operation, we close some facility $i' \in S$ and reconnect the clients in $\sigma^{-1}(i')$ to facilities in $S \setminus \{i'\}$. 
	\item In a $\swap$ operation, we open some facility $i \notin S$ and close some facility $i' \in S$, reconnect the clients in $\sigma^{-1}(i')$ to facilities in $S \setminus \{i'\} \cup \{i\}$, and possibly some other clients to $i$.   We say $i$ is \emph{swapped in} and $i'$ is \emph{swapped out} by the operation. 
\end{itemize}
Thus, in any valid operation, we can open and/or close at most one facility. A client can be reconnected if it is currently connected to the facility that will be closed, or it will be connected to the new open facility.  After we apply a local operation, $S$ and $\sigma$ will be updated accordingly so that $(S, \sigma)$ is always the current solution.

For the online algorithm with recourse model, since we need to bound the number of reconnections, we apply a local operation only if the \emph{scaled} cost it decreases is large compared to the number of reconnections it makes. This motivates the following definition:
\begin{definition}[Efficient operations for facility location]
\label{def:phieff}	
	Given a $\phi \geq 0$, we say a local operation on a solution $(S, \sigma)$ for a facility location instance is $\phi$-efficient, if it decreases $\cost_\lambda(S, \sigma)$ by more than $\phi$ times the number of clients it reconnects.
\end{definition}

The following two theorems can be derived from the analysis for the local search algorithms for facility location. We include their proofs in Appendix~\ref{appendix:local-search} for completeness.

\begin{restatable}{theorem}{uflapprox}\label{thm:FL-offline-apx-ratio}
	Consider a facility location instance with cost of the optimum solution being $\opt$ (using the original cost function).  Let $(S, \sigma)$ be the current solution in our algorithm and $\phi \geq 0$ be a real number. If there are no $\phi$-efficient local operations on $(S, \sigma)$, then we have 
	\begin{align*}
		\cost(S, \sigma) \leq \alpha_\FL\big(\opt + |C|\phi\big). 
	\end{align*} 
\end{restatable}
In particular, if we apply the theorem with $\phi = 0$, then we obtain that $(S, \sigma)$ is a $(\alpha_\FL = 1+\sqrt{2})$-approximation for the instance.

\sh{The following theorem will be used to analyze our randomized local search procedure.}
\begin{restatable}{theorem}{ufloperations}\label{thm:FL-offline-operations}
	Let $(S, \sigma)$  be a solution to a facility location instance and $\opt$ be the optimum cost. Then there are two sets $\calP_\rmC$ and $\calP_\rmF$ of valid local operations on $(S, \sigma)$, where each operation $\mathrm{op}$ decreases the scaled cost $\cost_\lambda(S, \sigma)$ by $\nabla_{\mathrm{op}} > 0$, such that the following holds: 
	\begin{itemize}
		\item $\sum_{\mathrm{op} \in \calP_\rmC} \nabla_{\mathrm{op}} \geq \cc(\sigma)- (\lambda f(S^*) + \cc(\sigma^*))  $.
		\item $\sum_{\mathrm{op} \in \calP_\rmF} \nabla_{\mathrm{op}} \geq  \lambda f(S) - (\lambda f(S^*) + 2\cc(\sigma^*)) $.
		\item There are at most $|F|$ $\close$ operations in $\calP_\rmC \biguplus \calP_\rmF$.
		\item For every $i \in F$,  there is at most 1 operation in each of $\calP_\rmC$ and $\calP_\rmF$ that opens or swaps in $i$.
	\end{itemize}
\end{restatable}

\subsection{Useful Lemmas}
The following lemmas will be used repeatedly in our analysis and thus we prove them separately in Appendix~\ref{appendix:helper-proofs}.
\begin{restatable}{lemma}{helpersumba}
	\label{lemma:helper-sum-b/a}
	Let $b \in \R_{\geq 0}^T$ for some integer $T \geq 1$. Let $B_{T'} = \sum_{t=1}^{T'} b_t$ for every $T' = 0, 1, \cdots, T$. Let $0 < a_1 \leq a_2 \leq \cdots \leq a_T$ be a sequence of real numbers and $\alpha > 0$ such that $B_t \leq \alpha a_t$ for every $t \in [T]$. Then we have 
	\begin{align*}
		\sum_{t = 1}^T \frac{b_t}{a_t} \leq \alpha \left(\ln \frac{a_T}{a_1} + 1\right).
	\end{align*}
\end{restatable}

\begin{restatable}{lemma}{helperstar}
	\label{lemma:helper-star}
	Assume at  some moment of an algorithm for facility location,  $C$ is the set of clients, $(S, \sigma)$ is the solution for $C$. Let $i \in F$ and $\tilde C \subseteq C$ be any non-empty set of clients.  Also at the moment there are no $\phi$-efficient operation that opens $i$ for some $\phi \geq 0$. Then we have 
	\begin{align*}
		d(i, S) \leq \frac{f_i + 2\sum_{\tilde j \in \tilde C} d(i, \tilde j)}{|\tilde C|} + \phi.
	\end{align*}
\end{restatable}

\paragraph{Organization} The rest of the paper is organized as follows. In Section~\ref{sec:ofl}, we prove Theorem~\ref{UFL-recourse} by giving our online algorithm for facility location with recourse. Section~\ref{sec:fast-UFL} gives the randomized local search procedure, that will be used in the proof of Theorem~\ref{UFL-dynamicIncremental} in Section~\ref{sec:dfl}.  Section~\ref{sec:dfl-fully} is dedicated to the proof of Theorem~\ref{UFL-fullydynamic}, by giving the fully dynamic algorithm for facility location in HST metrics. 
We give some open problems and future directions in Section~\ref{sec:discussions}. Some proofs are deferred to the appendix for a better flow of the paper.  
\section{$(1+\sqrt{2}+\epsilon)$-Competitive Online Algorithm with Recourse} \label{sec:ofl}
In this section, we prove Theorem~\ref{UFL-recourse} by giving the algorithm for online facility location with recourse. 
\subsection{The Algorithm}
For any $\epsilon >0$, let $\epsilon' = \Theta(\epsilon)$ be a parameter that is sufficiently small so that the approximation ratio $\alpha_\FL + O(\epsilon')= 1+\sqrt{2} + O(\epsilon')$ achieved by our algorithm is at most $\alpha_\FL + \epsilon$.  Our algorithm for online facility location is easy to describe. Whenever the client $j_t$ comes at time $t$, we use a simple rule to connect $j_t$, as defined in the procedure  $\mathsf{initial\mhyphen connect}$ in Algorithm~\ref{alg:initial-connect}: either connecting $j_t$ to the nearest facility in $S$, or opening and connecting $j_t$ to its nearest facility in $F \setminus S$, whichever incurs the smaller cost. Then we repeatedly perform $\phi$-efficient operations (Definition \ref{def:phieff}), until no such operations can be found, for $\phi=\frac{\epsilon'\cdot \cost(S, \sigma)}{\alpha_\FL|C|}$. \footnote{There are exponential number of possible operations, but we can check if there is a $\phi$-efficient one efficiently. $\close$ operations can be handled easily. To check if we can open a facility $i$, it suffices to check if $\sum_{j \in C: d(j, i) + \phi <  d(j,\sigma_j)} (d(j, \sigma_j) - d(j, i)- \phi ) > \lambda f_i \cdot 1_{i \notin S}$. $\swap$ operations are more complicated but can be handled similarly.}  

	\begin{algorithm}[htb]
		\caption{$\mathsf{initial\mhyphen connect}(j)$} \label{alg:initial-connect}
		\begin{algorithmic}[1]
				\If{$\min_{i \in F\setminus S}(f_i + d(i, j)) < d(j, S)$}
					\State let $i^* = \arg\min_{i \in F\setminus S}(f_i + d(i, j))$, $S \gets S \cup \{i^*\}, \sigma_j \gets i^*$
				\Else \ $\sigma_j \gets \arg\min_{i \in S} d(j, i)$
				\EndIf
		\end{algorithmic}
	\end{algorithm}

We can show that the algorithm gives an $(\alpha_\FL + \epsilon)$-approximation with amortized recourse $O(\log D\log n)$; recall that $D$ is the aspect ratio of the metric. To remove the dependence on $D$, we divide the algorithm into stages, and \emph{freeze} the connections of clients that arrived in early stages. The final algorithm is described in Algorithm~\ref{alg:ofl}, and Algorithm~\ref{alg:ofl-one-stage} gives  one stage of the algorithm. 

	\begin{algorithm}
			\caption{One Stage of Online Algorithm for Facility Location}
			\label{alg:ofl-one-stage}
			\begin{algorithmic}[1]
				\Require{
					\begin{itemize}
						\item $C$: initial set of clients
						\item $(S, \sigma)$: a solution for $C$ which is $O(1)$-approximate
						\item Clients $j_1, j_2, \cdots $ arrive from time to time
					\end{itemize}
				}
				\Ensure{
					Guaranteeing that $(S, \sigma)$ at the end of each time $t$ is $\frac{\alpha_\FL}{1 - \epsilon'}$-approximate 
				}
				\State $\mathsf{init} \gets \cost(S, \sigma)$
				\For{$t \gets 1, 2, \cdots$, terminating if no more clients will arrive}
					\State $C\gets C\cup\{j_t\}$, and call $\mathsf{initial\mhyphen connect}(j_t)$
					\label{step:ofl-settle-down}
					\While{there exists an $\frac{\epsilon'\cdot \cost(S, \sigma)}{\alpha_\FL|C|}$-efficient local operation} \label{step:ofl-while}
						perform the operation
					\EndWhile
					\If {$\cost(S, \sigma) > \mathsf{init}/\epsilon'$} terminate the stage \EndIf
				\EndFor
			\end{algorithmic}		
	\end{algorithm} 
	
	\begin{algorithm}[htb]
		\caption{Online Algorithm for Facility Location} \label{alg:ofl}
		\begin{algorithmic}[1]
			\State $C \gets \emptyset, S \gets \emptyset, \sigma = ()$
			\Repeat
				\State $C^\circ \gets C, (S^\circ, \sigma^\circ) \gets (S, \sigma)$
				\State redefine the next time to be time 1 and run one stage as defined in Algorithm \ref{alg:ofl-one-stage}
				\State permanently open one copy of each facility in $S^\circ$, and permanently connect clients in $C^\circ$ according to $\sigma^\circ$ (we call the operation \emph{freezing} $S^\circ$ and $C^\circ$)
				\State $C \gets C \setminus C^\circ$, restrict the domain of $\sigma$ to be the new $C$
			\Until no clients come
		\end{algorithmic}
	\end{algorithm}

In Algorithm~\ref{alg:ofl-one-stage}, we do as described above,\janr{what does before refers to here?}\shr{I changed to ``above''.} with two modifications. First, we are given an initial set $C$ of clients and a solution $(S, \sigma)$ for $C$ which is $O(1)$-approximate.  Second, the stage will terminate if the cost of our solution increases by a factor of more than $1/\epsilon'$. The main algorithm (Algorithm~\ref{alg:ofl}) is broken into many stages.  Since we shall focus on one stage of the algorithm for most part of our analysis, we simply redefine the time so that every stage starts with time 1.  The improved recourse comes from the \emph{freezing} operation: at the end of each stage, we permanently open one copy of each facility in $S^\circ$, and permanently connect clients in $C^\circ$ to copies of $S^\circ$ according to $\sigma^\circ$, where $C^\circ$ and $(S^\circ, \sigma^\circ)$ are the client set and solution at the beginning of the stage.   Notice that we assume the original facilities in $S^\circ$ will still participate in the algorithm in the future; that is, they are subject to opening and closing.  Thus each facility may be opened multiple times during the algorithm and we take the facility costs of all copies into consideration. This assumption is only for the sake of analysis; the actual algorithm only needs to open one copy and the costs can only be smaller compared to the described algorithm.  

From now on, we focus on one stage of the algorithm and assume that the solution given at the beginning of each stage is $O(1)$-approximate. In the end we shall account for the loss due to the freezing of clients and facilities.  Within a stage, the approximation ratio follows directly from Theorem~\ref{thm:FL-offline-apx-ratio}:  Focus on the moment after the while loop at time step $t$ in Algorithm~\ref{alg:ofl-one-stage}. Since there are no $\frac{\epsilon'\cdot \cost(S,\sigma)}{\alpha_\FL|C|}$-efficient local operations on $(S, \sigma)$, we have by the theorem that $\cost(S, \sigma) \leq \alpha_\FL\left(\opt + |C|\cdot \frac{\epsilon'\cdot \cost(S, \sigma)}{\alpha_\FL|C|}\right) = \alpha_\FL\opt + \epsilon'\cdot\cost(S, \sigma)$, where $\opt$ is the cost of the optimum solution for $C$.  Thus, at the end of each time, we have $\cost(S, \sigma) \leq \frac{\alpha_\FL}{1-\epsilon'}\cdot\opt$.

\subsection{Bounding Amortized Recourse in One Stage}
We then bound the amortized recourse in a stage; we assume that $\cost(S, \sigma) > 0$ at the beginning of the stage since otherwise there will be no recourse involved in the stage (since we terminate the stage when the cost becomes non-zero).  We use $T$ to denote the last time of the stage. For every time $t$, let $C_t$ be the set $C$ at the end of time $t$, and $\opt_t$ to be the cost of the optimum solution for the set $C_t$. For every $t \in [T]$, we define $\Delta_t$ to be the value of $\cost(S, \sigma)$ after Step~\ref{step:ofl-settle-down} at time step $t$ in Algorithm~\ref{alg:ofl-one-stage}, minus that before Step~\ref{step:ofl-settle-down}.  We can think of this as the cost increase due to the arrival of $j_t$.


The key lemma we can prove is the following:
\begin{lemma}\label{lmm:ofl-delta-cost-bound}
	For every $T' \in [T]$, we have
	$$\sum_{t = 1}^{T'} \Delta_t \leq O(\log T')\opt_{T'}.$$
\end{lemma}

\begin{proof}  Consider the optimum solution for $C_{T'}$ and focus on any star $(i, C')$ in the solution; that is, $i$ is an open facility and $C'$ is the set of clients connected to $i$.  Assume $C' \setminus C_0 = \{j_{t_1}, j_{t_2}, \cdots, j_{t_s}\}$, where $1 \leq t_1 < t_2 < \cdots < t_s \leq T'$; recall that $C_0$ is the initial set of clients given at the beginning of the stage. We shall bound $\sum_{s' = 1}^s\Delta_{t_{s'}}$ in terms of the cost of the star $(i, C' \setminus C_0)$.

By the rule specified in $\mathsf{initial\mhyphen connect}$, we have $ \Delta_{t_1 }\le f_i + d(i, j_{t_1})$. Now focus on any integer $k \in [2, s]$. Before Step~\ref{step:ofl-settle-down} at time $t_k$, no $\Big(\phi:= \frac{ \epsilon'\cdot  \cost(S, \sigma) }{\alpha_\FL|C_{t_k-1}|} \leq \frac{O(\epsilon')\cdot\opt_{t_k-1}}{t_k-1} \leq \frac{O(\epsilon')\cdot\opt_{T'}}{t_k-1} \Big)$-efficient operation that opens $i$ is available. Thus, we can apply Lemma~\ref{lemma:helper-star} on $i$, $\tilde C = \{j_{t_1}, j_{t_2}, \cdots, j_{t_{k-1}}\}$ and $\phi$ to conclude that before Step~\ref{step:ofl-settle-down}, we have
\begin{align*}
	d(i,S)\leq \frac{f_i + 2\cdot \sum_{k'=1}^{k-1 } d(i, j_{t_{k'} }) }{k-1}+ \frac{ O(\epsilon') \cdot  \opt_{T'}}{t_k-1}.
\end{align*}
%
%
In $\mathsf{initial\mhyphen connect}(j_{t_k})$, we have the option of connecting $j_{t_k}$ to its nearest open facility.  Thus, we have 
\begin{align*}
  \Delta_{t_{k} } \le d(i, S) + d(i, j_{t_{k}} ) 
  &\le  \frac{f_i + 2\cdot \sum_{k'=1}^{k-1 } d(i, j_{t_{k'} }) }{k-1}+ \frac{ O(\epsilon') \cdot  \opt_{T'}}{t_k-1 } + d(i, j_{t_{k}} ).
  \end{align*}
  
 We now sum up the above inequality for all $k \in [2, s]$ and that $\Delta_{t_1}\leq f_i + d(j, j_{t_1})$.  We get
 \begin{align}
 	\sum_{k=1}^s \Delta_{t_{k}} \leq O(\log s)\left(f_i + \sum_{k'=1}^sd(i, j_{t_{k'}})\right) + O(\epsilon')\sum_{k=2}^s\frac{\opt_{T'}}{t_k-1}. \label{inequ:each-star}
 \end{align}
 To see the above inequality, it suffices to consider the coefficients for $f_i$ and $d(i, j_{t_{k'}})$'s on the right-hand side.  The coefficient for $f_i$ is at most $1 + \frac11 + \frac12 + \cdots + \frac1{s-1} = O(\log s)$; the coefficient for each $d(i, j_{t_{k'}})$ is $ 1 + \frac{2}{k'} + \frac{2}{k'+1} + \cdots +\frac{2}{s-1} = O(\log s)$. 

We now take the sum of \eqref{inequ:each-star} over all stars $(i, C')$ in the optimum solution for $C_{T'}$.  The sum for the first term on the right side of \eqref{inequ:each-star} will be $O(\log T')\opt_{T'}$ since $f_i + \sum_{k'=1}^sd(i, j_{t_{k'}})$ is exactly the cost of the star $(i, C' \setminus C_0 \subseteq C')$. The sum for the second term will be $O(\epsilon'\log T')\cdot \opt_{T'}$ since the set of integers $t_k-1$ overall stars $(i, C')$ and all $k \geq 2$ are all positive and distinct. Thus overall, we have $\sum_{t = 1}^{T'} \Delta_t \leq O(\log T')\opt_{T'}$. 
\end{proof}

With Lemma~\ref{lmm:ofl-delta-cost-bound}, we can now bound the amortized recourse of one stage.  In time $t$, $\cost(S, \delta)$ first increases by $\Delta_t$ in Step~\ref{step:ofl-settle-down}. Then after that, it decreases by at least $\frac{\epsilon'\cost(S, \sigma)}{\alpha_\FL|C|} \geq \frac{\epsilon'\opt_t}{\alpha_\FL|C|} \geq \frac{\epsilon'\opt_t}{\alpha_\FL|C_T|}$ for every reconnection we made. 
 Let $\Phi_{T'} = \sum_{t = 1}^{T'}\Delta_t$; Lemma~\ref{lmm:ofl-delta-cost-bound} says $\Phi_t \leq \alpha \opt_{t}$ for some $\alpha = O(\log T)$ and every $t \in [T]$.  Noticing that $(\opt_t)_{t \in T}$ is a non-decreasing sequence, the total number of reconnections is at most
\begin{align*}	
	&\frac{\textsf{init}}{\epsilon'\cdot\opt_1/(\alpha_\FL|C_T|)} + \sum_{t=1}^T\frac{\Delta_t}{\epsilon' \cdot \opt_{t}/(\alpha_\FL|C_T|)} = \frac{\alpha_\FL|C_T|}{\epsilon'} \left( \frac{\textsf{init}}{\opt_1} + \sum_{t = 1}^{T-1} \frac{\Delta_t}{\opt_{t}} + \frac{\Delta_T}{\opt_T}\right).
\end{align*}
Notice that $\mathsf{init} \leq O(1)\opt_0 \leq O(1)\opt_1$.
Applying Lemma~\ref{lemma:helper-sum-b/a} with $T$ replaced by $T-1$, $b_t = \Delta_t, B_t = \Phi_t$ and $a_t = \opt_{t}$ for every $t$, we have that $\sum_{t=1}^{T-1}\frac{\Delta_t}{\opt_{t}} \leq \alpha \left(\ln\frac{\opt_{T-1}}{\opt_1} + 1\right) = O\left(\log T\log\frac1{\epsilon'}\right)$, since we have $\opt_{T-1} \leq O(1/\epsilon')\cdot\opt_1$\xguor{I guess this should be $\opt_{T-1} \leq \opt_1*O(1/\epsilon')$}\shr{Yes. Done}.  Notice that $\Delta_T \leq \opt_T$ since $\opt_T \geq \min_{i \in F}(f_i + d(i, j_T)) \geq \Delta_T$.  So, the total number of reconnections is at most $O\left(\frac{\log T}{\epsilon'}\log\frac1{\epsilon'}\right)\cdot|C_T|$. The amortized recourse per client is $O\left(\frac{\log T}{\epsilon'}\log\frac1{\epsilon'}\right) \leq O\left(\frac{\log n}{\epsilon'}\log\frac1{\epsilon'}\right)$, where in the amortization, we only considered clients  involved in the stage.   Recall that $n$ is the total number of clients arrived.

As each client appears in at most 2 stages, the overall amortized recourse is $O\left(\frac{\log n}{\epsilon'}\log\frac1{\epsilon'}\right)$. Finally we consider the loss in the approximation ratio due to freezing of clients.  Suppose we are in the $p$-th stage. Then the clients arrived at and before $(p-2)$-th stage has been frozen and removed.  Let $\overline{\opt}$ be the cost of the optimum solution for all clients arrived at or before $(p-1)$-th stage.  Then the frozen facilities and clients have cost at most $\overline\opt \cdot O\left(\epsilon' + \epsilon'^2 + \epsilon'^2 + \cdots \right) = O(\epsilon')\overline{\opt}$.  In any time in the $p$-th stage, the optimum solution taking all arrived clients into consideration has cost $\overline\opt' \geq \overline\opt$, and our solution has cost at most $(\alpha_\FL + O(\epsilon'))\overline\opt'$ without considering the frozen clients and facilities. Thus, our solution still has approximation ratio $\frac{(\alpha_\FL + O(\epsilon'))\overline\opt' + O(\epsilon')\overline\opt}{\overline\opt'}  = \alpha_\FL +  O(\epsilon')$ when taking the frozen clients into consideration. 
\section{Fast Local Search via Randomized Sampling} \label{sec:fast-UFL} 
%
\janr{This whole paragraph reads weird, and needs to rewritten. We have not defined category, the first line is incorrect.} \shr{Rewrote the paragraph.}
From now on, we will be concerned with dynamic algorithms. 
Towards proving Theorem \ref{UFL-dynamicIncremental} for the incremental setting, we first develop a randomized procedure that allows us to perform local search operations fast. 
In the next section, we use this procedure and ideas from the previous section to develop the dynamic algorithm with the fast update time.

The high level idea is as follows: We partition the set of local operations into many ``categories'' depending on which facility it tries to open or swap in. In each iteration of the procedure, we sample the category according to some distribution and find the best local operation in this category.  By only focusing on one category,  one iteration of the procedure can run in time $O(|C|\log |F|)$. On the other hand, the categories and the distribution over them are designed in such a way that in each iteration, the cost of our solution will be decreased by a multiplicative factor of $1 - \Omega\big(\frac1{ |F|}\big)$. This idea has been used in \cite{CharikarGhua2005} to obtain their $\tilde O(n^2)$ algorithm for approximating facility location. However, their algorithm was based on a different local search algorithm and analysis; for consistency and convenience of description, we stick to original local search algorithm of \cite{AryaGKMP01} that leads to $(1+\sqrt{2})$-approximation for the problem. Our algorithm needs to use the heap data structure.  



\subsection{Maintaining Heaps for Clients}
Unlike the online algorithm for facility location in Section~\ref{sec:ofl}, in the dynamic algorithm, we guarantee that the clients are connected to their nearest open facilities. That is, we always have $\sigma_j = \arg\min_{i\in S} d(j, i)$; we still keep $\sigma$ for convenience of description. We maintain $|C|$ min-heaps, one for each client $j \in C$: The min-heap for $j$ will contain the facilities in $S \setminus \{\sigma_j\}$, with priority value of $i$ being $d(j, i)$.   This allows us to efficiently retrieve the second nearest open facility to each $j$: This is the facility at the top of the heap for $j$ and we use the procedure $\mathsf{heap\mhyphen top}(j)$ to return it. 



\begin{figure*}
	\begin{algorithm}[H]
		\caption{$\mathsf{\Delta\mhyphen open}(i)$: \Return $\lambda f_i - \sum_{j \in C}  \max\{0, d(j, \sigma_{j}) - d(j, i)\}$}
		\label{alg:Delta-open}
	\end{algorithm}\vspace*{-25pt}
	\begin{algorithm}[H]
		\caption{$\mathsf{try\mhyphen open}(i)$} \label{alg:try-open}
		\begin{algorithmic}[1]
			\State \textbf{if} $\mathsf{\Delta\mhyphen open}(i) < 0$ \textbf{then} open $i$ by updating $S, \sigma$ and heaps accordingly
		\end{algorithmic}
	\end{algorithm}\vspace*{-25pt}
	\begin{algorithm}[H]
		\caption{$\mathsf{\Delta\mhyphen swap\mhyphen in}(i)$} \label{alg:Delta-swap-in}
		\begin{algorithmic}[1]
			\State $C' \gets \{j \in C: d(j, i) < d(j, \sigma_j)\}$ and $\Psi \gets \lambda f_i - \sum_{j \in C'} \big(d(j, \sigma_j) - d(j, i)\big)$ \label{step:Delta-swap-in-C'-Psi}
			\State $\Delta \gets \min_{i' \in S}\left\{\sum_{j \in \sigma^{-1}(i') \setminus C'}\big[\min\{d(j, i), d(j, \mathsf{heap\mhyphen top}(j))\} - d(j, i')\big] - \lambda f_{i'}\right\} + \Psi$ \label{step:Delta-swap-in-Delta}
			\State \Return $(\Delta, \text{the $i'$ above achieving the value of $\Delta$})$ \label{step:Delta-swap-in-i'}
		\end{algorithmic}
	\end{algorithm}\vspace*{-25pt}
	\begin{algorithm}[H]
		\caption{$\mathsf{\Delta\mhyphen close}$} \label{alg:Delta-close}
		\begin{algorithmic}[1]		
			\State $\Delta \gets \min_{i' \in S}\left\{\sum_{j \in \sigma^{-1}(i')}\big[d(j, \mathsf{heap\mhyphen top}(j)) - d(j, i')\big] - \lambda f_{i'}\right\}$
			\State \Return $(\Delta, \text{the $i'$ above achieving the value of $\Delta$})$
		\end{algorithmic}
	\end{algorithm}
\end{figure*}

We define four simple procedures $\mathsf{\Delta\mhyphen open}, \mathsf{try\mhyphen open},  \mathsf{\Delta\mhyphen swap\mhyphen in}$ and $\mathsf{\Delta\mhyphen close}$ that are described in Algorithms \ref{alg:Delta-open}, \ref{alg:try-open}, \ref{alg:Delta-swap-in} and \ref{alg:Delta-close} respectively.  Recall that we use the \emph{scaled cost}  for the local search algorithm; so we are working on the scaled cost function in all these procedures. $\mathsf{\Delta\mhyphen open}(i)$ for any $i \notin S$ returns $\Delta$, the increment of the scaled cost that will be incurred by opening $i$. (For it to be useful, $\Delta$ should be negative, in which case $|\Delta|$ indicates the cost decrement of opening $i$).  This is just one line procedure as in Algorithm~\ref{alg:Delta-open}; \janr{Which of these 1 line procedure?}\shr{in Algorithm~\ref{alg:Delta-open}} $\mathsf{try\mhyphen open}$ will open $i$ if it can reduce the scaled cost.  $\mathsf{\Delta\mhyphen swap\mhyphen in}(i)$ for some $i \notin S$ returns a pair $(\Delta, i')$, where $\Delta$ is the smallest scaled cost increment we can achieve by opening $i$ and closing some facility $i' \in S$, and $i'$ gives the facility achieving the smallest value. (Again, for $\Delta$ to be useful, it should be negative, in which case $i'$ is the facility that gives the maximum scaled cost decrement $|\Delta|$.)  Similarly, $\mathsf{\Delta\mhyphen close}$ returns a pair $(\Delta, i')$, which tells us the maximum scaled cost decrement we can achieve by closing one facility and which facility can achieve the decrement. Notice that in all the procedures, the facility we shall open or swap in is given as a parameter, while the facility we shall close is chosen and returned by the procedures.

With the heaps, the procedures $\mathsf{\Delta\mhyphen open}, \mathsf{\Delta\mhyphen swap\mhyphen in}$ and $\mathsf{\Delta\mhyphen close}$ can run in $O(|C|)$ time. We only analyze $\mathsf{\Delta\mhyphen swap\mhyphen in}(i)$ as the other two are easier.   First, we define $C'$ to be the set of clients $j$ with $d(j, i) < d(j, \sigma_j)$; these are the clients that will surely be reconnected to $i$ once $i$ is swapped in.  Let $\Psi = \lambda f_i - \sum_{j \in C'} (d(j, \sigma_j) - d(j, i))$ be the net scaled cost increase by opening $i$ and connecting $C'$ to $i$.  The computation of $C'$ and $\Psi$ in Step~\ref{step:Delta-swap-in-C'-Psi} takes $O(|C|)$ time.  If additionally we close some $i' \in S$, we need to reconnect each client in $\sigma^{-1}(i') \setminus C'$ to either $i$, or the top element in the heap for $j$, whichever is closer to $j$.  Steps \ref{step:Delta-swap-in-Delta} and \ref{step:Delta-swap-in-i'}  compute and return the best scaled cost increment and the best $i'$.  Since $\sum_{i' \in S}|\sigma^{-1}(i')| = |C|$, the running time of the step can be bounded by $O(|C|)$. 

The running time for $\mathsf{try\mhyphen open}$, swapping two facilities and closing a facility (which are not defined explicitly as procedures, but used in Algorithms~\ref{alg:sample}) can be bounded by $O(|C|\log |F|)$.  The running times come from updating the heap structures: For each of the $|C|$ heaps,  we need to delete and/or add at most $2$ elements; each operation takes time $O(\log  |F|)$.  


\subsection{Random Sampling of Local Operations}

\begin{figure*}
	\begin{algorithm}[H]
		\caption{$\mathsf{sampled\mhyphen local\mhyphen search}$} \label{alg:sample}
		\begin{algorithmic}[1]
				\If{$\mathsf{rand}(0, 1) < 1/3$} \Comment{$\mathsf{rand}(0, 1)$ returns a uniformly random number in $[0, 1]$}
					\State$(\Delta, i') \gets \mathsf{\Delta\mhyphen close}$
					\State\textbf{if} $\Delta < 0$ \textbf{then} close $i'$ by updating $S, \delta$ and heaps accordingly
				\Else
					\State $i \gets $ random facility in $F \setminus S$
					\State $\Delta \gets \mathsf{\Delta\mhyphen open}(i), (\Delta', i') \gets \mathsf{\Delta\mhyphen swap\mhyphen in}(i)$
					\State \textbf{if} $\Delta \leq \Delta'$ and $\Delta < 0$ \textbf{then} open $i$ by updating $S, \delta$ and heaps accordingly
					\State \textbf{else if} $\Delta' < 0$ \textbf{then} open $i$ and close $i'$ by updating $S, \delta$ and heaps accordingly
				\EndIf
		\end{algorithmic}
	\end{algorithm} \vspace*{-15pt}
	\begin{algorithm}[H]
		\caption{$\mathsf{FL\mhyphen iterate}(M)$} \label{alg:FL-iterate}
		\begin{algorithmic}[1]
				\State $(S^{\mathrm{best}}, \sigma^{\mathrm{best}}) \gets (S, \sigma)$
				\For{$\ell \gets 1$ to $M$}
					\State call $\mathsf{sampled\mhyphen local\mhyphen search}$
					\If{$\cost(S, \sigma) < \cost(S^{\mathrm{best}}, \sigma^{\mathrm{best}})$}  $(S^{\mathrm{best}}, \sigma^{\mathrm{best}}) \gets (S, \sigma)$ \EndIf  
				\EndFor
				\State \Return $(S^\best, \sigma^\best)$
		\end{algorithmic}
	\end{algorithm}
\end{figure*}

With the support of the heaps, we can design a fast algorithm to implement randomized local search.  $\mathsf{sampled\mhyphen local\mhyphen search}$ in Algorithm~\ref{alg:sample} gives one iteration of the local search. We first decide which operation we shall perform randomly. With probability $1/3$, we perform the $\close$ operation that will reduce the scaled cost the most (if it exists). With the remaining probability $2/3$,  we perform either an $\open$ or a $\swap$ operation. To reduce the running time, we randomly choose a facility $i \in F \setminus S$ and find the best operation that opens or swaps in $i$, and perform the operation if it reduces the cost. One iteration of $\mathsf{sampled\mhyphen local\mhyphen search}$ calls the procedures in Algorithms~\ref{alg:Delta-open} to \ref{alg:Delta-close} at most once and performs at most one operation, and thus has running time $O(|C|\log  |F|)$.

In the procedure $\mathsf{FL\mhyphen iterate}(M)$ described in Algorithm~\ref{alg:FL-iterate}, we run the $\mathsf{sampled\mhyphen local\mhyphen search}$ $M$ times.  It returns the best solution obtained in these iterations, according to the \emph{original (non-scaled) cost}, which is not necessarily the solution given in the last iteration. So we have
\begin{obs}
	\label{obs:time-iterate}
	The running time of $\mathsf{FL\mhyphen iterate}(M)$ is $O(M|C|\log |F|)$, where $C$ is the set of clients when we run the procedure.
\end{obs}

Throughout this section, we fix a facility location instance. Let $(S^*, \sigma^*)$ be the optimum solution (w.r.t the original cost) and $\opt = \cost(S^*, \sigma^*)$ be the optimum cost.  Fixing one execution of $\mathsf{sampled\mhyphen local\mhyphen search}$,  we use $(S^0, \sigma^0)$ and $(S^1, \sigma^1)$ to denote the solutions before and after the execution respectively. Then, we have 
\begin{restatable}{lemma}{samplelocalsearch} \label{lemma:sample-local-search}
	Consider an execution of $\mathsf{sampled\mhyphen local\mhyphen search}$ and fix $(S^0, \sigma^0)$. We have 
		\begin{align*}
			\cost_\lambda(S^0, \sigma^0) - \E[\cost_\lambda(S^1, \sigma^1)] \geq \frac1{3 |F|}\max\left\{
				\begin{array}{c}
					\cc(\sigma^0) - (\lambda f(S^*) + \cc(\sigma^*))\\
					\lambda f(S) - (\lambda f(S^*) + 2\cc(\sigma^*))\\
					\cost_\lambda(S^0, \sigma^0) - (2\lambda f(S^*) + 3\cc(\sigma^*))
				\end{array}
			\right\}.
		\end{align*}	
\end{restatable}

\begin{restatable}{lemma}{fliterate}
	\label{lemma:ufl-iterate}
	Let $(S^\circ, \sigma^\circ)$ be the $(S, \sigma)$ at the beginning of an execution of $\mathsf{FL\mhyphen iterate}(M)$, and assume it is an $O(1)$-approximation to the instance. Let $\Gamma \geq 2$ and $M = O\left(\frac{ |F|}{\epsilon'}\log\Gamma\right)$ is big enough. Then with probability at least $1-\frac1\Gamma$, the solution returned by the procedure is $(\alpha_\FL + \epsilon')$-approximate.
\end{restatable}



	\section{$(1+\sqrt{2}+\epsilon)$-Approximate Dynamic Algorithm for Facility Location in Incremental Setting}
\label{sec:dfl}

In this section, we prove Theorem~\ref{UFL-dynamicIncremental} by combining the ideas from Sections  \ref{sec:ofl} and \ref{sec:fast-UFL} to derive a dynamic algorithm for facility location in the incremental setting. As for the online algorithm in Section~\ref{sec:ofl}, we divide our algorithm into stages. Whenever a client comes, we use a simple rule to accommodate it.  Now we can not afford to consider all possible local operations as in Section~\ref{sec:ofl}. Instead we use the randomized local search idea from the algorithm in Section~\ref{sec:fast-UFL} by calling the procedure $\mathsf{FL\mhyphen iterate}$. We call the procedure only if the cost of our solution has increased by a factor of $1+\epsilon'$ (where $\epsilon' = \Theta(\epsilon)$ is small enough).  In our analysis, we show a lemma similar to Lemma \ref{lmm:ofl-delta-cost-bound}: The total increase of costs due to arrival of clients is small, compared to the optimum cost for these clients. Then, we can bound the number of times we call $\mathsf{FL\mhyphen iterate}$. Recall that we are given an integer $\Gamma = \poly\big(n, \log D, \frac1\epsilon\big)$ that is big enough: We are aiming at a success probability of $1-1/\Gamma$ for each call of $\mathsf{FL\mhyphen iterate}$. Our final running time will only depend on $O(\log \Gamma)$.   

The main algorithm will be the same as Algorithm~\ref{alg:ofl}, except that we use  Algorithm~\ref{alg:fast-UFL-one-stage} as the algorithm for one stage.  As before, we only need to design one stage of the algorithm. Recall that in a stage we are given an initial set $C$ of clients, an $O(1)$-approximate solution $(S, \sigma)$ for $C$.  Clients come one by one and our goal is to maintain an $(\alpha_\FL +  O(\epsilon'))$-approximate solution at any time. The stage terminates if no client comes or our solution has cost more than $1/\epsilon'$ times the cost of the initial solution. 

\begin{algorithm}
	\caption{One Stage of Dynamic Algorithm for Facility Location} \label{alg:fast-UFL-one-stage}
	\begin{algorithmic}[1]
		\Require{ 
			\begin{itemize}
				\item $C$: the initial set of clients
				\item $(S, \sigma)$: initial solution for $C$, which is $O(1)$-approximate
			\end{itemize}
		}
		\State let $M = O\left(\frac{ |F|}{\epsilon'}\log\Gamma\right)$ be large enough \label{step:fast-UFL-M}
		\State $(S, \sigma) \gets \mathsf{FL\mhyphen iterate}\left(M\right)$, $\textsf{init}\gets \cost(S, \sigma),  \last \gets \textsf{init}$ \label{step:fast-UFL-init}
		\For{$t \gets 1, 2, 3, \cdots$, terminating if no more clients arrive}
			\For{$q = \ceil{\log\frac{\last}{|F|}}$ to $\ceil{\log\frac\last{\epsilon'}}$} \label{step:fast-UFL-enumerate-q}
				\State \textbf{if} $i \gets \arg\min_{i \in F\setminus S, f_i \leq 2^q}d(j_t, i)$ exists, \textbf{then} call $\mathsf{try\mhyphen open'}(i)$  \label{step:fast-UFL-try-open} \Comment{$\mathsf{try\mhyphen open'}$ is the same as $\mathsf{try\mhyphen open}$ except we consider the cost instead of scaled cost.}
			\EndFor
			\State $C \gets C \cup \{j_t\}$ and call $\mathsf{try\mhyphen open'}\big(\arg\min_{i \in F \setminus S}(d(j_t, i) + f_i)\big)$
			\label{step:fast-UFL-handle-j}
			\If{$\cost(S, \sigma) > (1+\epsilon')\cdot \last$}  
					\State $(S, \sigma) \gets \mathsf{FL\mhyphen iterate}\left(M\right)$ \label{step:fast-UFL-call-iterate}
					\If {$\cost(S, \sigma) > \last$} $\last \gets \cost(S, \sigma)$  \EndIf  \label{step:fast-UFL-update-last}
					\If{$\last > \textsf{init}/\epsilon'$}  terminate the stage \EndIf \label{step:fast-UFL-terminate}
			\EndIf
		\EndFor
	\end{algorithmic}
\end{algorithm}

%
Notice that in a stage, we are considering the original costs of solutions (instead of scaled costs as inside $\mathsf{FL\mhyphen iterate}$).  During a stage we maintain a value $\last$ which gives an estimation on the cost of the current solution $(S, \sigma)$.  Whenever a client $j_t$ comes, we apply some rules to open some facilities and connect $j_t$ (Steps~\ref{step:fast-UFL-enumerate-q} to \ref{step:fast-UFL-handle-j}). These operations are needed to make the cost increase due to the arrival of $j_t$ (defined as $\Delta_t$ later) small.   In the algorithm $\mathsf{try\mhyphen open'}$ is the same as $\mathsf{try\mhyphen open}$, except that we use the original cost instead of the scaled cost (this is not important but only for the sake of convenience). If $\cost(S, \sigma)$ becomes too large, i.e, $\cost(S, \sigma) > (1+\epsilon')\last$,  then we call $(S, \sigma) \gets \mathsf{FL\mhyphen iterate}(M)$ for the $M$ defined in Step~\ref{step:fast-UFL-M} (Step~\ref{step:fast-UFL-call-iterate}), and update $\last$ to $\cost(S, \sigma)$ if we have $\cost(S, \sigma) > \last$ (Step~\ref{step:fast-UFL-update-last}). We terminate the algorithm when $\last \geq \mathsf{init}/\epsilon$, where $\mathsf{init}$ is $\cost(S, \sigma)$ at the beginning of the stage (Step~\ref{step:fast-UFL-terminate}).  

We say an execution of $\mathsf{FL\mhyphen iterate}(M)$ is successful if the event in Lemma~\ref{lemma:ufl-iterate} happens. Then we have
\begin{lemma}
	\label{lemma:ufl-dynamic-ratio}
	If all executions of $\mathsf{FL\mhyphen iterate}$ are successful, the solution $(S, \sigma)$ at the end of each time is  $(1+\epsilon')(\alpha_\FL+\epsilon')$-approximate. 
\end{lemma}
\begin{proof}
	 This holds since we always have $\cost(S, \sigma) \leq (1+\epsilon')\last$ at the end of each time, where $\last$ is the cost of some $(\alpha_\FL +  \epsilon')$-approximate solution at some moment before.  As we only add clients to $C$, the cost of the optimum solution can only increase and thus the claim holds. 
\end{proof}

Now we argue each execution of $\mathsf{FL\mhyphen iterate}(M)$ is successful with probability at least $1-1/\Gamma$.  This will happen if $(S, \sigma)$ is $O(1)$-approximate before the call.  By Lemma~\ref{lemma:ufl-iterate}, we only need to make sure that the $(S, \sigma)$ before the execution is $O(1)$-approximate. This is easy to see: Before Step~\ref{step:fast-UFL-handle-j} in time $t$, we have $\cost(S, \sigma) \leq O(1)\opt$; the increase of $\cost(S, \sigma)$ in the step is at most the value of $\opt$ after the step (i.e, we consider the client $j_t$ when defining $\opt$). Thus, we have $\cost(S, \sigma) \leq O(1)\opt$ after the step.


\subsection{Bounding Number of Times of Calling $\mathsf{FL\mhyphen iterate}$}
It remains to bound the number of times we call $\mathsf{FL\mhyphen iterate}$. Again, we use $T$ to denote the last time step of Algorithm~\ref{alg:fast-UFL-one-stage} (i.e, one stage of the dynamic algorithm) and $\Delta_t$ to denote the cost increase due to the arrival of $j_t$:  it is the value of $\cost(S, \sigma)$ before Step~\ref{step:fast-UFL-handle-j} minus that after Step~\ref{step:fast-UFL-handle-j} in time $t$. For every time $t \in [T]$, let $C_t$ be the set $C$ at the end of time $t$, and let $\opt_t$ be the cost of the optimum solution for $C_t$.  Let $\last_t$ be the value of $\last$ at the \emph{beginning} of time $t$.  

Due to Step~\ref{step:fast-UFL-handle-j}, we have the following observation:
\begin{obs}
	\label{obs:dfl-delta-t}
	For every $t \in [T]$, we have $\Delta_t \leq \min_{i \in F}(f_i + d(i, j_t))$.
\end{obs}
\begin{proof}
	Let $i = \arg\min_{i \in F}(f_i + d(i, j_t))$ and consider Step~\ref{step:fast-UFL-handle-j} at time $t$. If $d(j_t, S) \leq f_i + d(i, j_t)$ before the step, then we have $\Delta_t \leq d(i, j_t)$. Otherwise, $i \notin S$ and $d(j_t, S) > f_i + d(i, j_t)$. Then $\mathsf{try\mhyphen open'}(i)$ in the step will open $i$ and we have $\Delta_t \leq f_i + d(i, j_t)$.
\end{proof}

We can also prove the following lemma that bounds $\Delta_t$:
\begin{lemma}
	\label{lemma:dfl-delta-t}
	Let $t \in [T], i^* \in F$ such that $f_{i^*} \leq \last_t/\epsilon'$ and $C' \subseteq C_{t-1}$ be any non-empty subset.  Then we have 
	\begin{align*}
		\Delta_t \leq \frac{2}{|C'|}\left(\max\set{f_{i^*}, \last_t/|F|} + \sum_{j \in C'}d(i^*, j)\right) + 5d(i^*, j_t).
	\end{align*}
\end{lemma}
\begin{proof}
	In this proof, we focus on the time $t$ of the algorithm. If $i^* \in S$ before Step~\ref{step:fast-UFL-handle-j}, then we have $\Delta_t \leq d(i^*, j_t)$ and thus we can assume $i^* \notin S$ before Step~\ref{step:fast-UFL-handle-j}.  Since Loop~\ref{step:fast-UFL-enumerate-q} only adds facilities to $S$, we have that $i^* \notin S$ at any moment in Loop~\ref{step:fast-UFL-enumerate-q}.
	
	Let $q = \ceil{\log \max\set{f_{i^*}, \last_t/|F|}}$; notice this $q$ is considered in Loop~\ref{step:fast-UFL-enumerate-q}.  Let $i \in F\setminus S$ be the facility with $f_i \leq 2^q$ nearest to $j_t$ at the beginning of the iteration $q$; this is the facility we try to open in Step~\ref{step:fast-UFL-try-open} in the iteration for $q$.  Notice that $d(j_t, i) \leq d(j_t, i^*)$ since $i^*$ is a candidate facility.
	
	 Since we called $\mathsf{try\mhyphen open}(i)$ in Step~\ref{step:fast-UFL-try-open}, there is no $0$-efficient opening operation that opens $i$ after the step. Then, we can apply Lemma~\ref{lemma:helper-star} on  this facility $i$, the set $C'$ and $\phi = 0$. So, after Step~\ref{step:fast-UFL-try-open} of the iteration for $q$, we have
	\begin{align*}
		d(j_t, S) \leq \frac{1}{|C'|}\left(f_i + 2\sum_{j \in C'}d(i, j)\right)  + d(i, j_t).
	\end{align*}
	
	Notice that  $d(i, i^*) \leq d(i, j_t) + d(j_t, i^*) \leq 2d(j_t, i^*)$,  $f_i \leq 2\max\set{f_{i^*}, \epsilon'\last_t/|F|} $ and $S$ can only grow before the end of Step~\ref{step:fast-UFL-handle-j}. We have 
	\begin{align*}
		\Delta_t &\leq \frac{1}{|C'|}\left(2\max\set{f_{i^*},\last_t/|F|} + 2\sum_{j \in C'}(d(i^*, j) + d(i^*, i))\right) + d(i^*, j_t) \\
		&\leq \frac{2}{|C'|}\left(\max\set{f_{i^*},\last_t/|F|} + \sum_{j \in C'}d(i^*, j)\right) + 5d(i^*, j_t). \qedhere
	\end{align*}
\end{proof}
		
With the lemma, we can then prove the following lemma:
\begin{lemma}
	\label{lemma:dfl-Delta}
	For every $T' \in [T-1]$, we have
	\begin{align*}
		\sum_{t = 1}^{T'} \Delta_t \leq O(\log T') \cdot \opt_{T'}
	\end{align*}
\end{lemma}
\begin{proof}
	The proof is similar to that of Lemma~\ref{lmm:ofl-delta-cost-bound}. Let $(S^*, \sigma^*)$ be the optimum solution for clients $C_{T'}$.  Focus on some $i^* \in S^*$ and assume $(C_{T'} \setminus C_0) \cap \sigma^{*-1}(i^*) = \{j_{t_1}, j_{t_2}, \cdots, j_{t_s}\}$  with $1 \leq t_1 < t_2 < \cdots < t_s \leq T'$. 
	
	We have $\Delta_{t_1} \leq f_{i^*} + d(i^*, j_{t_1})$ by Observation~\ref{obs:dfl-delta-t}. Then focus on any $k \in [2, s]$. If $f_{i^*} > \last_{t_k}/\epsilon$, then we must have $\opt_{t_k} \geq \last_{t_k}/\epsilon$ and the stage will terminate at time ${t_k}$. Thus ${t_k} = T$, contradicting the assumption that ${t_k} \leq T' \leq T-1$.  So we assume $f_{i^*} \leq \last_{t_k}/\epsilon$. We can apply Lemma~\ref{lemma:dfl-delta-t} with $i^*$ and $C' = \{j_{t_1}, j_{t_2}, \cdots, j_{t_{k-1}} \}$ to obtain that $\Delta_{t_k} \leq \frac{2}{k-1}\left(\max\set{f_{i^*},\last_{t_k}/|F|} + \sum_{k'=1}^{k-1}d(i^*, j_{t_{k'}})\right) + 5d(i^*, j_{t_k})$.  We can replace $\last_{t_k}$ with $\last_{T'}$ since $\last_{t_k} \leq \last_{T'}$.  
	
	The sum of upper bounds over all $k \in [s]$ is a linear combinations of $\max\set{f_{i^*},\last_{T'}/|F|}$ and $d(i^*, j_{t_{k'}})$'s.  
	In the linear combination, the coefficient for $\max\set{f_{i^*},\last_{T'}/|F|}$ is at most $1 + \frac21 + \frac22 + \frac23 +   \cdots + \frac2{s-1} = O(\log s)  = O(\log T')$. The coefficient for $d(i^*, j_{t_{k'}})$ is at most $5 + \frac2{k'} + \frac2{k'+1} + \cdots \frac2{s-1} = O(\log s) = O(\log T')$. Thus, overall, we have $\sum_{k = 1}^{s}\Delta_{t_k} \leq O(\log T') \big(\max\set{f_{i^*},\last_{T'}/|F|} + \sum_{k'=1}^s d(i^*, j_{t_{k'}})\big)$.
	
	Therefore $\sum_{t = 1}^{T'}  \Delta_t \leq O(\log T') \left( \cost(S^*, \sigma^*) + |S^*|\last_{T'}/|F|\right)$, by taking the sum of the above inequality over all $i^* \in S^*$. The bound is at most $O(\log T')(\opt_{T'} + \last_{T'}) = O(\log T') \cdot \opt_{T'}$,  since $|S^*| \leq |F|$ and $\last_{T'} \leq O(1)\opt_{T'-1} \leq O(1) \opt_{T'}$.
\end{proof}


Between two consecutive calls of $\mathsf{FL\mhyphen iterate}$ in Step~\ref{step:fast-UFL-call-iterate} at time $t_1$ and $t_2 > t_1$, $\cost(S, \sigma)$ should have increased by at least $\epsilon'\last_{t_2}$: At the end of time $t_1$, we have $\cost(S, \sigma) \leq \last_{t_1+1} = \last_{t_2}$ since otherwise $\last$ should have been updated in time $t_1$.  We need to have $\cost(S, \sigma) > (1+\epsilon')\last_{t_2}$ after Step~\ref{step:fast-UFL-handle-j} at time $t_2$ in order to call $\mathsf{FL\mhyphen iterate}$. Thus, the increase of the cost during this period is at least $\epsilon' \last_{t_2}$.   Thus, we have $\sum_{t=t_1+1}^{t_2}\frac{\Delta_t}{\epsilon'\cdot\last_t} \geq 1$ since $\last_t = \last_{t_2}$ for every $t \in (t_1, t_2]$.  The argument also holds when $t_1 = 0$ and $t_2 > t_1$ is the first time in which we call $\mathsf{FL\mhyphen iterate}$.  Counting the call of $\mathsf{FL\mhyphen iterate}$ in Step~\ref{step:fast-UFL-init}, we can bound the total number of times we call the procedure by $1 + \frac{1}{\epsilon'}\sum_{t=1}^T\frac{\Delta_t}{\last_t}$. 

Again let $\Phi_{T'}= \sum_{t = 1}^{T'} \Delta_t$ for every $T' \in [0, T]$. Lemma~\ref{lemma:dfl-Delta} says $\Phi_{t} \leq O(\log t) \opt_{t}$ for every $t \in [0, T-1]$. For every $t \in [T]$, since $\Delta_t \leq \opt_t$, thus we have $\Phi_t  = \Phi_{t-1} + \Delta_t \leq O(\log t) \opt_{t-1} \leq O(\log T) \last_t$ since $\last_t$ will be at least the cost of some solution for $C_{t-1}$.  Applying Lemma~\ref{lemma:helper-sum-b/a} with $a_t = \last_t, b_t = \Delta_t$ and $B_t = \Phi_t$ for every $t$, the number of times we call $\mathsf{FL\mhyphen iterate}$ can be bounded by 
\begin{align*}
	1+\frac{1}{\epsilon'}\sum_{t=1}^T\frac{\Delta_t}{\last_t} \leq \frac{1}{\epsilon'} O(\log T) \left(\ln\frac{\last_T}{\last_1}  + 1\right) = O\left(\frac{\log T}{\epsilon}\log\frac{1}{\epsilon}\right).
\end{align*}

We can then analyze the running time and the success probability of our algorithm. Focus on each stage of the algorithm. By Observation~\ref{obs:time-iterate}, each call to $\mathsf{FL\mhyphen iterate}(M)$ takes time $O(M|C|\log |F|) = O\left(\frac{ |F|}{\epsilon'}(\log \Gamma) |C|\log n \right) = O\left(\frac{ n\cdot|C_T|}{\epsilon}\log^2 n\right)$, where $C$ is the set of clients in the algorithm at the time we call the procedure, $C_T \supseteq C$ is the number set of clients at the end of time $T$, and $M = O\left(\frac{|F|}{\epsilon'}\log \Gamma\right)$ is as defined in Step~\ref{step:fast-UFL-M}.  The total number of times we call the procedure is at most $O\left(\frac{\log T}{\epsilon}\log\frac1\epsilon\right) \leq O\left(\frac{\log n}{\epsilon}\log\frac1\epsilon\right)$.  Thus, the running time we spent on $\mathsf{FL\mhyphen iterate}$ is $O\left(\frac{ n\cdot|C_T|}{\epsilon^2}\log^3 n\log\frac{1}{\epsilon}\right)$.  The running time for Steps~\ref{step:fast-UFL-enumerate-q} to \ref{step:fast-UFL-handle-j} is at most $T \cdot O\big(\log \frac{|F|}{\epsilon'}\big) \cdot O\big(|C_T|\log |F|\big) = O(|C_T|T\log^2 \frac{|F|}{\epsilon}) \leq O(n|C_T|\log^2\frac{n}{\epsilon})$. Thus, the total running time of a stage is at most $O\left(\frac{ n\cdot|C_T|}{\epsilon^2}\log^3 n\log\frac{1}{\epsilon}\right)$. Now consider all the stages together. The sum of $|C_T|$ values over all stages is at most $2n$ since every client appears in at most 2 stages. So, the total running time of our algorithm is $O\left(\frac{n^2}{\epsilon^2}\log^3 n\log\frac1\epsilon\right)$.  

For the success probability, the total number of times we call $\mathsf{FL\mhyphen iterate}(M)$ is at most $O\left(\log_{1/\epsilon} (nD)\frac{\log n}{\epsilon}\log \frac1\epsilon\right) = \poly(\log n, \log D, \frac1\epsilon)$. If we have  $\Lambda$ is at least $n^2$ times this number, which is still $\poly(n, \log D, \frac{1}{\epsilon})$, then the success probability of our algorithm is at least $1-1/n^2$.   

Finally, we remark that the success of the algorithm only depends on the success of all executions of $\mathsf{FL\mhyphen iterate}$.  Each execution has success probability $1-1/\Gamma$ even if the adversary is adaptive.  This finishes the proof of Theorem~\ref{UFL-dynamicIncremental}.

\paragraph{Remark} We can indeed obtain an algorithm that has both $O(\log T)$ amortized client recourse and $\tilde O(n^2)$ total running time, by defining $\phi = \frac{\cost(S, \sigma)}{\alpha_\FL\epsilon'}$ and only performing $\phi$-efficient local operations.  
However, this will require us to put $\phi$ everywhere in our analysis and deteriorate the cleanness of the analysis.  Thus, we choose to separate the two features in two algorithms: small recourse and $\tilde O(n^2)$ total running time. 

We also remark that the total running time for all calls of $\mathsf{FL\mhyphen iterate}$ is only $\tilde O(n|F|)$, and the $\tilde O(n^2)$ time comes from Steps~\ref{step:fast-UFL-enumerate-q} to \ref{step:fast-UFL-handle-j}.  By losing a multiplicative factor of $2$ and additive factor of $1$ in the approximation ratio, we can assume every client is collocated with its nearest facility (See Appendix~\ref{appendix:moving-clients}). Then at any time we only have $O(|F|)$ different positions for clients, and the running time of the algorithm can be improved to $O(\frac{n|F|}{\epsilon^2}\log^3n\log\frac1{\epsilon})$. 
\section{Fully Dynamic Algorithm for Facility Location on Hierarchically Well Separated Tree Metrics}
\label{sec:dfl-fully}

%
In this section, we give our fully dynamic algorithm for facility location on hierarchically-well-separated-tree (HST) metrics. Our algorithm achieves $O(1)$-approximation and $O(\log^2D)$ amortized update time.   As we mentioned early, we assume each client is collocated with a facility.  From now on, we fix the HST $T$ and assume the leaves of $T$ is $X = F$; let $V$ be the set of all nodes in $T$. Let $d_T$ be the metric induced by $T$ over the set $V$ of vertices.

\noindent{\bf Notations.} Recall that $\level(v)$ is the level of $v$ in $T$. For every vertex $v \in V$, define $\Lambda_v$ to be the set of children of $v$, $X_v$ to be the set of leaf descendants  of $v$, and $T_v$ be the maximal sub-tree of $T$ rooted at $v$. We extend the facility cost from $X$ to all vertices in $V$: for every $v \in V \setminus X$, we define $f_v = \min_{i \in X_v}f_i$.  We can assume that each internal vertex $v$ is a facility; by opening $v$ we mean opening a copy of the $i \in X_v$ with $f_i = f_v$. This assumption only loses a factor of $2$ in the competitive ratio: On one hand, having more facilities can only make our problem easier; on the other hand, the cost of connecting a client to any $i \in X_v$ is at most twice that of connecting it to $v$. By the definition, the facility costs along a root-to-leaf path are non-decreasing.

\subsection{Offline Algorithm for Facility Location on HST Metrics}
In this section, we first give an offline $O(1)$-approximation algorithm for facility location on the HST metric $d_T$ as a baseline.  Notice that facility location on trees can be solved exactly using dynamic programming. However the algorithm is hard to analyze in the dynamic algorithm model since the solution is sensitive to client arrivals and departures.  Our algorithm generalizes the algorithm in \cite{EsencayiGLW19} for facility location with uniform facility cost,  that was used  to achieve the differential privacy requirement. 

For every vertex $v \in V$, we let $N_v$ be the number of clients at locations in $X_v$. Although according to the definition $N_v$'s are integers, in most part of the analysis we assume there are non-negative \emph{real numbers}. This will be useful when we design the dynamic algorithm.  Let $\alpha \in \{1, 2\}^V$ and $\beta \in \{1, 2\}^{V \setminus X}$ be vectors given to our algorithm.  They are introduced solely for the purpose of extending the algorithm to the dynamic setting; for the offline algorithm we can set $\alpha$ and $\beta$ to be all-1 vectors. 

\paragraph{Marked and Open Facilities} For every vertex $v \in V$, we say $v$ is \emph{marked} w.r.t the vectors $N$ and $\alpha$ if $$N_v \cdot 2^{\level(v)} > f_v/\alpha_v $$ and \emph{unmarked} otherwise. 
The following observation can be made:
\begin{obs}
	Let $u$ be the parent of $v$. If $v$ is marked w.r.t $N$ and $\alpha$, so is $u$.
\end{obs}
\begin{proof}
	$v$ is marked w.r.t $N$ and $\alpha$ implies $N_v 2^{\level(v)} > f_v/\alpha_v $.  Notice that $N_u \geq N_v, \level(u) = \level(v) + 1,  \alpha_v \leq 2\alpha_u$ and $f_u \leq f_v$. So, $N_u 2^{\level(u)} \geq 2N_v2^{\level(v)} > 2 f_v/\alpha_v \geq f_u/\alpha_u $.
\end{proof}
Thus there is a monotonicity property on the marking status of vertices in $T$. 
We say a vertex $v$ is highest unmarked (w.r.t $N$ and $\alpha$) if it is unmarked and its parent is marked; we say a vertex $v$ is lowest marked if it is marked but all its children are unmarked.  However, sometimes we say a vertex $u$ is the lowest marked ancestor of a leaf $v \in X$ if either $u=v$ is marked, or $u\neq v$ is marked and the child of $u$ in the $u$-$v$ path is unmarked; notice that in this case, $u$ might not be a lowest marked vertex since it may have some other marked children. If we need to distinguish between the two cases, we shall use that $u$ is lowest marked \emph{globally} to mean $u$ is a lowest marked vertex.

If a leaf vertex $v \in X$ is marked, then we open $v$. For every marked vertex $v \in V\setminus X$, we open $v$ if and only if  $$\left(\sum_{u \in \Lambda_v: u \text{ unmarked}} N_u  \right)2^{\level(v)}> f_v/(\alpha_v\beta_v).$$ 
Notice that all unmarked vertices are closed.

\begin{obs}
	\label{obs:departure-lowest-open}
	If $v$ is lowest marked, then $v$ is open.
\end{obs}
\begin{proof}
	 We can assume $v \notin X$ since otherwise $v$ is open. So, $N_v 2^{\level(v)} > f_v/\alpha_v$ and all children of $v$ are unmarked. Thus, $\sum_{u \in \Lambda_v: {u\text{ unmarked}}}N_u = \sum_{u \in \Lambda_v}N_u = N_v$. Therefore, $\left(\sum_{u \in \Lambda_v: {u\text{ unmarked}}}N_u\right) 2^{\level(v)} = N_v 2^{\level(v)} > f_v/\alpha_v \geq f_v/(\alpha_v\beta_v)$. Thus $v$ will be open. 
\end{proof}

With the set of open facilities defined, every client is connected to its nearest open facility according to $d_T$, using a consistent tie-breaking rule (e.g, the nearest open facility with the smallest index).  We assume the root $r$ of $T$ has $\frac{f_v}{2^{\level(v)}} \leq 1$ by increasing the number of levels. So $r$ will be marked whenever $N_r \geq 1$.   This finishes the description of the offline algorithm. 

\paragraph{Analysis of $O(1)$-Approximation Ratio.} We show the algorithm achieves an $O(1)$-approximation.  First we give a lower bound on the optimum cost. For every $v \in V$,  let $$\LB(v) = \min\set{N_v2^{\level(v)}, f_v}.$$ Then we have
\begin{lemma} \label{lemma:departure-LB}
	Let $U$ be a set of vertices in $T$ without an ancestor-descendant pair; i.e, for every two distinct vertex $u$ and $v$ in $U$, $u$ is not an ancestor of $v$. Then the cost of the optimum solution is at least $\sum_{v \in U}\LB(v)$. 
\end{lemma}
\begin{proof}
	Fix an optimum solution. Consider any $v \in U$. We consider the cost inside $T_v$ in the optimum solution: the connection cost of clients, plus the cost of open facilities in $T_v$. Then this cost is at least $\LB(v)= \min\set{N_v2^{\level(v)}, f_v}$: If we open a facility in $T_v$ then the facility cost is at least $f_v$; otherwise, all the $N_v$ clients in $T_v$ have to be connected to outside $T_v$, incurring a cost of at least $N_v2^{\level(v)}$. The lemma follows from that the trees $T_v$ over all $v \in U$ are disjoint and thus we are not over-counting the costs in the optimum solution.
\end{proof}

Then let $U$ be the set of highest unmarked vertices and marked leaves; clearly $U$ does not have an ancestor-descendant pair. By Lemma~\ref{lemma:departure-LB}, the optimum cost is at least $\sum_{v \in U}\LB(v)$. We prove the following lemma.
\begin{lemma}
	\label{lemma:departure-UB}
	 The solution produced by our algorithm has cost at most $O(1)\sum_{u \in U}\LB(u)$.
\end{lemma}

\begin{proof}
First consider the facility cost of our solution. If a leaf $v$ is marked and open, we have $N_v > f_v/\alpha_v$ (as $\level(v) = 0$) and thus $\LB(v) = \min\set{N_v,f_v} \geq f_v/\alpha_v$. Then $f_v$ can be bounded by $\alpha_v\LB(v) \leq 2\LB(v)$.   If $v \in V \setminus X$ is marked and open, then by our algorithm we have $\sum_{u \in \Lambda_v: u \text{ unmarked}}N_u 2^{\level(v)} \geq f_v/(\alpha_v\beta_v )$. Since each $u$ in the summation is unmarked, we have $\LB(u) = N_u 2^{\level(u)}$. Thus, we have $\sum_{u \in \Lambda_v: u\text{ unmarked}}\LB(u)  = \frac12\sum_{u}N_u 2^{\level(v)} \geq \frac12 f_v/(\alpha_v\beta_v) \geq \frac18 f_v$.  That is $f_v$ can be bounded by $8\sum_{u \in \Lambda_v:u \text{ unmarked}}\LB(u)$. Notice that each $u$ in the summation has $u \in U$ since it is highest unmarked.  So, summing the bounds over all open facilities $v$ gives us that the facility cost of our solution is at most $8\sum_{u \in U}\LB(u)$. 

Now consider the connection cost. For every $v \in X$, let $u$ be the highest unmarked ancestor of $v$ (if $v$ itself is open, then its connection cost is $0$ and we do not need to consider this case). Let $w$ be the parent of $u$; so $w$ is marked. Then there must be an open facility in the maximal tree rooted at $w$: consider any lowest marked vertex in the sub-tree rooted at $w$; it must be open by Lemma~\ref{obs:departure-lowest-open}.  Thus, any client at $v$ has connection cost at most $2 \times 2^{\level(w)} = 4 \times 2^{\level(u)}$. Thus, the total connection cost in our solution is at most $4\sum_{u \in U \setminus X}N_u2^{\level(u)} = 4\sum_{u \in U \setminus X}\LB(u)$. This finishes the proof of the lemma. 
\end{proof}
Combining Lemmas~\ref{lemma:departure-LB} and \ref{lemma:departure-UB} gives that our algorithm is an $O(1)$-approximation.  One lemma that will be useful in the analysis of dynamic algorithm is the following:
\begin{lemma}
	\label{lemma:departure-ub-outside}
	For any open facility $v$ in our solution, the number of clients connected to $v$ that are outside $T_v$ is at most $O(\log D)\frac{f_v}{2^{\level(v)}}$.
\end{lemma} 
\begin{proof}
	We consider each ancestor $u$ of $v$ and count the number clients connected to $v$ with lowest common ancestor with $v$ being $u$. Focus on a child $w$ of $u$ that is not $v$ or an ancestor of $v$. If $w$ is marked, then no clients in $T_w$ will be connected to $v$ since some facility in $T_w$ will be open.  Thus, let $U'$ be the unmarked children of $u$ that is not $v$ or an ancestor of $v$.  Then if we have $\sum_{w \in U'}N_w2^{\level(u)} \geq f_u/(\alpha_u\beta_u)$, then $u$ will be marked and open and clients in $T_w, w \in U'$ will not be connected to $v$. Otherwise we have $\sum_{w \in U'}N_w< f_u/(\alpha_u\beta_u\cdot 2^{\level(u)} ) \leq f_u/2^{\level(u)} \leq f_v/2^{\level(v)}$ as $f_u \leq f_v$ and $\level(u) \geq \level(v)$.  The lemma follows since we have at most $O(\log D)$ ancestors of $v$.
\end{proof}


\paragraph{Remark} The algorithm so far gives a \emph{data structure} that supports the following operations in $O(\log D)$ time: i) updating $N_v$ for some $v \in X$ and ii) returning the nearest open facility of a leaf $v \in X$.    Indeed the algorithm can be made simpler: We set $\alpha$ to be the all-1 vector, and we open the set of lowest marked facilities (so both $\alpha$ and $\beta$ are not needed).  For every vertex $u \in V$, we maintain the nearest open facility $\psi_u$ to $u$ in $T_u$.  Whenever a client at $v$ arrives or departs, we only need change $N_u$, $\psi_u$, marking and opening status of $u$ for ancestors $u$ of $v$.  To return the closest open facility to a leaf $v \in X$, we travel up the tree from $v$ until we find an ancestor $u$ with $\psi_u$ defined, and return $\psi_u$.  Both operations take $O(\log D)$ time.  However, our goal is to maintain the solution $(S, \sigma)$ \emph{explicitly} in memory. Thus we also have to bound the the number of reconnections during the algorithm, since that will be a lower bound on the total running time.

\subsection{Dynamic Algorithm for Facility Location on HST Metrics}
In this section, we extend the offline algorithm to a dynamic algorithm with $O(\log^3 D)$-amortized update time; recall that $D$ is the aspect ratio of the metric. We maintain $\alpha, \beta$ and $N$-vectors, and at any moment of the algorithm, the marking and opening status of vertices are exactly the same as that obtained from the offline algorithm for $\alpha, \beta$ and $N$.

Initially, let $\alpha$ and $\beta$ be all-$1$ vectors, and $N$ be the all-0 vector.  So all the vertices are unmarked. 
Whenever a client at some $v \in X$ arrives or departs, the $\alpha, \beta$ values, the marking and opening status of ancestors of $v$ may change and we show how to handle the changes. The vertices that are not ancestors of $v$ are not affected during the process.

When a client at $v$ arrives or departs, we increase or decrease the $N_u$ values for all ancestors $u$ of $v$ by 1 \emph{continuously} at the same rate (we can think of that the number of clients at $v$ increases or decreases by 1 continuously).  During the process, the marking and opening status of these vertices may change. If such an event happens, we change $\alpha$ and/or $\beta$ values of the vertex so that it becomes harder for the status to change back  in the future. Specifically, we use the following rules:
\begin{itemize}
	\item If a vertex $u$ changes to marked (from being unmarked), then we change $\alpha_u$ to $2$ (notice that $u$ remains marked w.r.t the new $\alpha$), and $\beta_u$ to $1$. In this case, we do not consider the opening status change of $u$ as an event.
	\item If a vertex $u$ changes to unmarked (from being marked), we change $\alpha_u$ to $1$ (notice that $u$ remains unmarked w.r.t the new $\alpha$). The $\beta_u$ value becomes useless. In this case, we also do not consider the opening status change of $u$ as an event.
	\item If a marked vertex $u$ becomes open (from being closed), then we change $\beta_u$ to $2$ (notice that $u$ remains open w.r.t the new $\beta$).
	\item If a marked vertex $u$ becomes closed (from being open), then we change $\beta_u$ to $1$ (notice that $u$ remains closed w.r.t the new $\beta$).
\end{itemize}
We call the 4 types of events above as marking, unmarking, opening and closing events.  

Now we talk about the order the events happen. When we increase $N_u$ values of ancestors of $v$ continuously, one of the following two events may happen:
\begin{itemize}
	\item The highest unmarked ancestor $u$ of $v$ may become globally lowest marked, and this may \emph{induce} a closing event for the parent $w$ of $u$.
	\item The lowest marked ancestor $u$ of $v$ may become open.
\end{itemize}
Similarly, when we decrease $N_u$ values of ancestors of $v$ continuously, one of the following two events may happen:
\begin{itemize}
	\item The lowest marked ancestor $u$ of $v$ may become unmarked (we must that $u$ was lowest marked globally), and this may \emph{induce} an opening event for the parent $w$ of $u$.
	\item The lowest marked ancestor $u$ of $v$ may become closed. 
\end{itemize}
Above, if two events happen at the same time, we handle an arbitrary event. Notice that after we handle the event, the conditions for the other event might not hold any more, in which case we do not handle it.

Once we have finished the process of increasing or decreasing $N_u$ values by 1, the clients will be connected to their respective nearest open facilities, breaking ties using the consistent rule.  A reconnection happens if a client is connected to a different facility. 

\paragraph{Bounding Number of Reconnections} Now we analyze the reconnections made in the algorithm.  When a client at $v \in X$ arrives or departs, at most $O(\log D)$ vertices $u$ will have their $N_u$ values changed by $1$. We distribute 4 tokens to each ancestor $u$ of $v$, that are of type-A, type-B, type-C and type-D respectively.\footnote{The types are only defined for convenience.}  We are going to use these tokens to charge the events happened.


First focus on the  sequence of marking/unmarking events happened at a vertex $u$.  Right before $u$ becomes unmarked we have $N_u \leq f_u/(2 \times 2^{\level(u)})$ since at the moment we have $\alpha_u = 2$.  Immediate after that $\alpha_u$ is changed to $1$. For $u$ to become marked again, we need $N_u \leq f_u/2^{\level(u)}$. So during the period $N_u$ must have been increased by at least $f_u/(2 \times 2^{\level(u)})$.  Similarly, right before $u$ becomes marked we have $N_u \geq f_u/2^{\level(u)}$ since at the moment  we have $\alpha_u = 1$. Then we change $\alpha_u$ to $2$ immediately. For $u$ to become unmarked again, $N_u$ should be decreased by at least $ f_u/(2\times2^{\level(u)})$.  So, when a marking/unmarking event happens at $u$, we can spend $\Omega(f_u/2^{\level(u)})$ type-A tokens owned by $u$.

Then we focus on the sequence $\calS$ of opening/closing events at $u$ between two adjacent marking/unmarking events at $u$. At these moments, $u$ is marked and $\alpha_u = 2$.  For the first event in $\calS$, we can spend $\Omega(f_u/2^{\level(u)})$ type-B tokens owned by $u$.  If some opening/closing event $e$ in $\calS$ is induced by an unmarking/marking event of some child $u'$ of $u$, then we can spend $\Omega(f_{u'}/2^{\level(u')}) \geq \Omega(f_u/2^{\level(u)})$ type-C tokens owned by $u'$ for $e$, and the event $e'$ after $e$ in $\calS$ if it exists. Notice that we already argued that $u'$ has collected enough number of type-C tokens.  

Then we focus on an event $e'$ in $\calS$ such that both $e$ and the event $e$ before $e'$ in $\calS$ are not induced. First, assume $e$ is an opening event and $e'$ is a closing event. Then, after $e$ we have $\sum_{u' \in \Lambda_u: u' \text{ unmarked}} N_{u'} = f_u/(2 \times 2^{\level(u)})$ and before $e'$ we have $\sum_{u' \in \Lambda_u: u' \text{ unmarked}} N_{u'} = f_u/(4 \times 2^{\level(u)})$.   Notice that the set of unmarked children of $u$ may change, and let $U'$ and $U''$ be the sets of unmarked children of $u$ at the moments after $e$ and before $e'$ respectively. Again if there is some $u' \in (U' \setminus U'') \cup (U'' \setminus U')$, we spend $\Omega(\frac{f_{u'}}{2^{\level(u')}}) \geq \Omega(\frac{f_u}{2^{\level(u)}})$ type-C tokens owned by $u'$.  Otherwise, $U = U'$ and $f_u/(4\times 2^{(\level(u))})$ clients in $T_u$ must have departed between $e$ and $e'$ and we can then spend $\Omega(f_u/2^{\level(u)})$ type-D tokens for $e'$.  The case when $e$ is an closing event and $e'$ is an opening event can be argued in the same way. 

Thus, whenever an event happens at $u$, we can spend $\Omega(f_u/2^{\level(u)})$ tokens; moreover if an opening/closing event at $u$ was induced by an unmarking/marking event at some child $u'$ of $u$, then we can spend $\Omega(f_{u'}/2^{\level(u')})$ tokens for the event at $u$. A facility $u$ changes its opening status when an event happens at $u$.    Notice that, we reconnect a client  only if it was connected to a ready-to-close facility, or it needs to be connected  to newly open facility.  By Lemma~\ref{lemma:departure-ub-outside}, at any moment the number of clients connected to $u$ from outside $T_u$ is at most $O(\log D)\cdot \frac{f_u}{2^{\level(u)}}$. At the moment $u$ changes its opening status because of an non-induced event, then before and after the event the number of clients connected to $u$ from $T_u$ is of order $O\left(\frac{f_u}{2^{\level(u)}}\right)$.   $u$ changes its opening status due to a marking/unmarking event happened at some child $u'$ of $u$, then before and after the event the number of clients connected to $u$ from $T_u$ is of order $\Theta\left(\frac{f_{u'}}{2^{\level(u')}}\right)$.  Thus, on average, for each token we spent we connect at most $O(\log D)$ clients. Since each client arrival or departure distributes at most $O(\log D)$ tokens, we have that the amortized number of reconnections (per client arrival/departure) is at most $O(\log^2D)$.

\paragraph{Analyzing Update Time} Then with the bound on the number of reconnections (recourse), we can bound the update time easily. Indeed, we can maintain a $\psi_u$ for every $u \in V$, which indicates the nearest open facility to $u$ in $T_u \setminus u$ ($\psi_u$ could be undefined). We also maintain a value $N'_u$ for marked vertices $u$ where $N'_u = \sum_{v \in \Lambda_v, v\text{ unmarked}} N_v$.  Whenever a client at $v$ arrives or departs, we need to change $\alpha_u, \beta_u, N_u, N'_u, \psi_u$, marking and opening status of $u$ only for ancestors $u$ of $v$. The update can be made in $O(\log D)$ time for every client arrival or departure using the information on the vertices. The bottleneck of the algorithm comes from reconnecting clients. We already argued that the amortized number of reconnections per client arrival/departure is $O(\log^2D)$ and thus it suffices to give an algorithm that can find the clients to be connected efficiently.  

For every vertex $u$, we maintain a double-linked-list of unmarked children $u'$ of $u$ with $N_{u'} \geq 1$.  With this structure it is easy to see that for every client that needs to be reconnected, we need $O(\log D)$ time to locate it.  If $u$ becomes open, we need to consider each unmarked children $u'$ of $u$ and reconnect clients in $T_{u'}$ to $u$. The time needed to locate these clients can be made $O(\log D)$ times the number of clients.   For every strict ancestor $w$ of $u$ for which there are no open facilities in between we can use the $\psi_w$ information to see if we need to reconnect clients in $T_w$. If yes, then for every unmarked child $w'$ of $w$ with $N_{w'} \geq 1$ that is not an ancestor of $u$, we need to connect the clients in $T_{w'}$ to $u$. Again enumerating these clients takes time $O(\log D)$ times the number of clients.  Similarly, if $u$ becomes closed, we then need to connect all clients connected to $u$ to the nearest open facility to $u$, which can be computed using $\psi$ values of $u$ and its ancestors.  Enumerating the clients takes time $O(\log D)$ times the number of clients.  Overall, the amortized running time per client arrival/departure is $O(\log ^3D)$.
\section{Open Problems and Discussions} \label{sec:discussions}

We initiated the study of facility location problem in general metric spaces in recourse and dynamic models.
Several interesting problems remain open: The most obvious one is can we get $O(1)$-competitive online/dynamic algorithms
with polylog amortized recourse or fast update times in the fully dynamic setting.  Another interesting direction is can we extend our results to the capacitated facility location and capacitated $k$-median, where there is an upper bound on the number of clients that can be assigned to a single open facility.
From technical point of view, it would be interesting to find more applications of local search and probabilistic tree embedding techniques in the dynamic algorithms model.
Finally, as alluded in the introduction, an exciting research direction is to understand the power of recourse in the online model.
	\bibliographystyle{plain}
\bibliography{reflist} 
	\appendix
	\section{Analysis of Offline Local Search Algorithms for Facility Location}
\label{appendix:local-search}

In this section, we prove theorems related to the local search algorithms for facility location.
\subsection{Local Search for facility location}
\uflapprox*
\begin{proof}
	This proof is almost identical to the analysis of the $\alpha_FL$-approximation local search algorithm for facility location, except we take $\phi$ into consideration in all the inequalities. Eventually we shall have an $\alpha_FL|C|\phi$ term on the right side of the inequality.
		
	Formally, we let $(S^*, \sigma^*)$ be the optimum solution to facility location instance. Focus on an $i^* \in S^*$. Since there is no $\phi$-efficient operation that opens $i^*$ (recall that we can open $i^*$ even if we have $i\in S^*$), we have
	   \begin{align*}
	     \sum_{j \in \sigma^{*-1}(i^*) }d(j,\sigma_j) \le  \lambda f_{i^* } \cdot 1_{i^*\notin S}
	     + \sum_{j \in \sigma^{*-1}(i^*)} (d(j,i^*) + \phi)
	   .\end{align*}
	   This implies
	   \begin{align}
	     \sum_{j \in \sigma^{*-1}(i^*) } d(j,\sigma_j) \le  \lambda f_{i^*} 
	     + \sum_{j \in \sigma^{*-1}(i^*)} d(j,i^*) + |\sigma^{*-1}(i^*)|\phi. \label{inequ:ufl-open}
	   \end{align}
	Summing the inequalities over all $i^* \in S^*$ gives us
	\begin{align}
		\cc(\sigma) \leq  \lambda f(S^*) + \cc(\sigma^*) + |C|\phi. \label{inequ:ufl-C}
	\end{align}
	
	For every $i \in S$, let $\psi(i)$ be the nearest facility in $S^*$ to $i$.  For every $i^* \in S^*$ with $\psi^{-1}(i^*) \neq \emptyset$, let $\psi^*(i^*)$ be the nearest facility in $\psi^{-1}(i^*)$ to $i^*$.   
	
	Focus on some $i \in S, i^* = \psi(i)$ such that $\psi^*(i^*) = i$. The operation that swaps in $i^*$, swaps out $i$ and connects $\sigma^{*-1}(i^*) \cup \sigma^{-1}(i)$ to $i^*$ is not $\phi$-efficient. This implies
	  \begin{align*}
	     &\quad  \lambda f_i + \sum_{j \in \sigma^{*-1}(i^*) \cup \sigma^{-1}(i)} d(j,\sigma_j) \\
	     & \le   \lambda f_{i^*} + \sum_{j \in \sigma^{*-1}(i^*)}d(j, i^*) 
	     + \sum_{j \in \sigma^{-1}(i) \setminus \sigma^{*-1}(i^*)}d(j,  i^*)
	     + \big|\sigma^{*-1}(i^*) \cup \sigma^{-1}(i)\big|\phi \\
	     & \le  \lambda f_{i^*} + \sum_{j \in \sigma^{*-1}(i^*)} d(j,i^*) +
	     \sum_{j \in \sigma^{-1}(i) \setminus \sigma^{*-1}(i^*)}
	     [d(j,\sigma^*(j)) + 2d(j,i)]
	     + \big|\sigma^{*-1}(i^*) \cup \sigma^{-1}(i)\big|\phi.
	   \end{align*}
	   To see the second inequality, notice that $d(j, i^*) \leq d(j, i) + d(i, i^*) \leq d(j, i) + d(i, \sigma^*(j)) \leq 2d(j, i) + d(j, \sigma^*(j))$. Canceling  $\sum_{j \in \sigma^{-1}(i) \setminus \sigma^{*-1}(i^*)} d(j, i)$ on both sides and relaxing the right side a bit gives us
	   \begin{align}
			\quad  \lambda f_i + \sum_{j \in \sigma^{*-1}(i^*)} d(j,\sigma_j)&\leq  \lambda f_{i^*} + \sum_{j \in \sigma^{*-1}(i^*)}d(j, i^*)  + \big|\sigma^{*-1}(i^*) \cup \sigma^{-1}(i)\big|\phi \nonumber\\
			 &+ \sum_{j \in \sigma^{-1}(i)} \left(d(j,i) + d(j, \sigma^*(j)))\right). \label{inequ:ufl-swap}
	   \end{align}
	   Notice that it could happen that $i = i^*$ in the above setting; the inequality was implied by the operation that opens $i = i^*$ and connects $\sigma^{*-1}(i^* = i)$ to $i$. 
	   
	   Now, focus on a $i \in S$ with $\psi^*(\psi(i)) \neq i$.  Then closing $i$ and connecting each client in $j \in \sigma^{-1}(i)$ to $\psi^*(\sigma^*(j)) \neq i$  is not $\phi$-efficient.  So, we have 
	   \begin{align*}
	      \lambda f_i + \sum_{j \in \sigma^{-1}(i)} d(j,i) &\le  \lambda f_i+ \sum_{j \in \sigma^{-1}(i)}d(j, \psi^*(\sigma^*(j)) )
	     + \big|\sigma^{-1}(i)\big|\phi \\
	     & \le 
	     \sum_{j \in \sigma^{-1}(i)}
	     [2d(j,\sigma^*(j)) + d(j,i)]
	     + \big|\sigma^{-1}(i)\big|\phi. 
	   \end{align*}
	   To see the inequality, we have $d(j, \psi^*(\sigma^*(j))) \leq d(j, \sigma^*(j)) + d(\sigma^*(j), \psi(\sigma^*(j))) \leq d(j, \sigma^*(j)) + d( \sigma^*(j), i) \leq 2d(j, \sigma^*(j)) + d(j, i)$.
	   This implies 
	   \begin{align}
	   	 \lambda f_i \leq 2\sum_{j \in \sigma^{-1}(i)}
	   	     d(j,\sigma^*(j)) 
	   	     + \big|\sigma^{-1}(i)\big|\phi.  \label{inequ:ufl-close}
	   \end{align}
	   
	   Now, consider the inequality obtained by summing up \eqref{inequ:ufl-swap} for all pairs $(i, i^*)$ with $i^* = \psi(i)$ and $\psi^*(i^*) = i$, \eqref{inequ:ufl-close} for all $i$ with $\psi^*(\psi(i)) \neq i$, and \eqref{inequ:ufl-open} for all $i^*$ with $\psi^{-1}(i^*) = \emptyset$. This inequality will be $ \lambda f(S) + \cc(\sigma) \leq  \lambda f(S^*) + 2\cc(\sigma^*) + \cc(\sigma) + 2|C|\phi$, which is 
	   \begin{align}
	   		 \lambda f(S) \leq  \lambda f(S^*) + 2\cc(\sigma^*) + 2|C|\phi. \label{inequ:ufl-F}
	   \end{align}
	   Summing up Inequalities~\eqref{inequ:ufl-C} and $1/\lambda$ times \eqref{inequ:ufl-F} gives $f(S) + \cc(\sigma) \leq (1+\lambda) f(S^*) + (1+2/\lambda)\left(\cc(\sigma^*) + |C|\phi\right) = \alpha_FL\left(\opt + |C|\phi\right)$, since $1 + \lambda = 1+2/\lambda = 1+\sqrt{2}=\alpha_FL$. This finishes the proof of Theorem~\ref{thm:FL-offline-apx-ratio}.
  \end{proof}
  
\ufloperations*

The theorem follows from the proof of Theorem~\ref{thm:FL-offline-apx-ratio}. Let $\phi = 0$ in the theorem statement and the proof.  \eqref{inequ:ufl-C} and \eqref{inequ:ufl-F} were obtained by adding many of the inequalities of the form \eqref{inequ:ufl-open}, \eqref{inequ:ufl-swap} and \eqref{inequ:ufl-close}. 
Notice that each inequality corresponds to a local operation. In the setting for Theorem~\ref{thm:FL-offline-operations}, the inequalities do not hold anymore since we do not have the condition that $0$-efficient operations do not exist. However for an inequality correspondent to an operation $\textrm{op}$,  we can add $\nabla_\text{op}$ to the right side so that the inequality becomes satisfied.  Then adding all the inequalities that were used to obtain \eqref{inequ:ufl-C},  we obtain
\begin{align*}
	\cc(\sigma) \leq \lambda f(S^*) + \cc(\sigma^*) + \sum_{\textrm{op} \in \calP_\rmC} \nabla_{\textrm{op}}
\end{align*}
where $\calP_\rmC$ is the set of operations correspondent to the inequalities. Similarly we can obtain a set $\calP_\rmF$ of operations, such that
\begin{align*}
	\lambda f(S) \leq \lambda f(S^*) + 2\cc(\sigma^*) + \sum_{\textrm{op} \in \calP_\rmF} \nabla_{\textrm{op}}.
\end{align*}

It is easy to check that each of $\calP_\rmC$ and $\calP_\rmF$ contains at most 1 operation opens or swaps in $i^*$, for every $i^* \in S^* \subseteq f$ and does not contain operations that open or swap in facilities outside $S^*$. $\calP_\rmC \uplus \calP_\rmF$ contains at most $|S| \leq |F|$ close operations. Rewriting the two inequalities almost gives us Theorem~\ref{thm:FL-offline-operations}, except for the requirement that each $\textrm{op} \in \calP_\rmC \cup \calP_\rmF$ has $\nabla_{\mathrm{op}} > 0$; this can be ensured by removing $\textrm{op}$'s with $\nabla_{\textrm{op}} \leq 0$ from $\calP_\rmC$ and $\calP_\rmF$.

\section{Proofs of Useful Lemmas} \label{appendix:helper-proofs}
\helpersumba*
\begin{proof} Define $a_{T+1} = +\infty$.
	\begin{align*}
		\sum_{t = 1}^T \frac{b_t}{a_t} &= \sum_{t = 1}^T \frac{B_t - B_{t-1}}{ a_{t}}=\sum_{t = 1}^{T} B_t \left(\frac{1}{a_t} - \frac{1}{a_{t+1}}\right) = \sum_{t = 1}^{T}\frac{B_t}{a_t} \left(1 - \frac{a_t}{a_{t+1}}\right) \leq \alpha \sum_{t = 1}^{T}\left(1 - \frac{a_{t}}{a_{t+1}}\right)\\
		&=\alpha T- \alpha\sum_{t = 1}^{T-1}\frac{a_t}{a_{t+1}} \leq \alpha T- \alpha(T-1)\Big(\frac{a_1}{a_T}\Big)^{1/(T-1)} \\
		&=  \alpha(T-1)\left(1-e^{-\ln\frac{a_T}{a_1}/(T-1)}\right) + \alpha \leq \alpha(T-1)\ln\frac{a_T}{a_1}/(T-1) + \alpha = \alpha\left(\ln \frac{a_T}{a_1}+1\right).
	\end{align*}
	The inequality in the second line used the following fact: if the product of $T-1$ positive numbers is $\frac{a_1}{a_T}$, then their sum is minimized when they are equal.  The inequality in the third line used that $1-e^{-x} \leq x$ for every $x$.
\end{proof}

\helperstar*
\begin{proof}
	By the conditions in the lemma, opening facility $i$ and reconnecting  $\tilde C$ to $i $ is not $\phi$-efficient. This gives that at the moment, we have
	\[
	    \sum_{\tilde j \in \tilde C}d(\tilde j, S) \leq \sum_{\tilde j \in \tilde C}d(\tilde j, \sigma_{\tilde j})
	    \leq  f_i + \sum_{\tilde j \in \tilde C} d(i , \tilde j)
	    + |\tilde C|\cdot \phi
	      \]
	By triangle inequalities we have $ d(\tilde j, S) \ge d(i, S) - d(i, \tilde j )$ for every $\tilde j \in \tilde C$. Combining with the previous inequality yields:
	\begin{align*}
		d(i, S) \le  \frac{1}{|\tilde C|}\sum_{\tilde j \in \tilde C}\left(d(\tilde j, S) + d(i, \tilde j)\right) \leq \frac{ f_i + 2\sum_{\tilde j \in \tilde C} d(i, \tilde j) }{|\tilde C|}+ \phi. \hspace*{80pt} \qedhere
	\end{align*}
\end{proof}

\section{Moving Clients to Facility Locations}
\label{appendix:moving-clients}
	In this section we show that by moving clients to their nearest facilities, we lose a multiplicative factor of $2$ and an additive factor of $1$ in the approximation. That is, an $\alpha$ approximate solution for the new instance, is $2\alpha+1$ approximate for the original instance. Throughout this section, we simply use the set of open facilities to define a solution and all clients are connected to their respective nearest open facilities.
	
	Let a facility location instance be given by $F, (f_j)_{j \in C}, C$ and $d$. Let $\psi_j$ be the nearest facility in $F$ to $j$ for every $j \in C$.  By moving all clients $j$ to $\psi_j$, we obtain a new instance. Let $S^*$ be the optimum solution to the original instance. Suppose we have an solution $S$ for the new instance that is $\alpha$-approximate solution. Thus $f(S) + \sum_{j \in C}d(\psi_j, S) \leq \alpha\left(f(S^*) + \sum_{j \in C}d(\psi_j, S^*)\right)$. We show that $S$ is $2\alpha+1$ approximate for the original instance.
	
	Notice that for every $j \in C$, we have $d(j, S) - d(j, \psi_j ) \leq  d(\psi_j, S) \leq d(j, S) + d(j, \psi_j)$ by triangle inequalities.
	\begin{align*}
		f(S) + \sum_{ j \in C} d(j, S)  &\leq f(S) + \sum_{ j \in C} \left(d(\psi_j, S) + d(j, \psi_j)\right) \\
		&\leq \alpha\left(f(S^*) + \sum_{j \in C}d(\psi_j, S^*)\right) + \sum_{j \in C} d(j, \psi_j)
	\end{align*}
	For every $j \in C$, since $\psi_j$ is the nearest facility in $F$ to $j$, we have $d(\psi_j, S^*) \leq d(j, \psi_j) + d(j, S^*) \leq 2d(j, S^*)$. Thus, we have 
	\begin{align*}
		f(S) + \sum_{ j \in C} d(j, S) &\leq \alpha f(S^*) + 2\alpha\sum_{j \in C}d(j, S^*) + \sum_{j \in C} d(j, \psi_j)\\
		&\leq \alpha f(S^*) + (2\alpha + 1)\sum_{j \in C}d(j, S^*).
	\end{align*}
	Thus, we have that $S$ is a $(2\alpha+1)$-approximate solution for the original instance.

\section{Missing Proofs from Section~\ref{sec:fast-UFL}}
\samplelocalsearch*
\begin{proof}
		 We are going to lower bound the expected value of $\cost_\lambda(S^0, \sigma^0) - \cost_\lambda(S^1, \sigma^1)$.  By Theorem~\ref{thm:FL-offline-operations}, there are two sets $\calP_\rmC$ and $\calP_\rmF$ of local operations satisfying the properties.  Below, we let $\calQ$ be one of the following three sets: $\calP_\rmC$, or $\calP_\rmF$, or $\calP_\rmC \biguplus \calP_\rmF$.
		 
		 For every $i \in F$, let $\calQ_{i}$ be the set of operations in $\calQ$ that open or swap in $i$. Let $\calQ_{\emptyset}$ be the set of $\close$ operations in $\calQ$.  Let $\Phi_i$ be maximum of $\nabla_{\mathsf{op}}$ over all $\mathsf{op} \in \calQ_{i}$ (define $\Phi_i = 0$ if $\calQ_{i} = \emptyset$); define $\Phi_\emptyset$ similarly. Notice that if $i \in S$ then open $i$ will not decrease the cost since we maintain that all the clients are connected to their nearest open facilities. Thus, $\calQ_{i} = \emptyset$ for $i \in S$.  Then, conditioned on that we consider $\close$ operations in $\mathsf{sampled\mhyphen local\mhyphen search}$, the cost decrement of the iteration is at least $\Phi_\emptyset$. Conditioned on that we consider opening or swapping in $i$ in the iteration, the decrement is at least $\Phi_i$. Thus, 
		  $\cost_\lambda(S^0, \sigma^0) - \E[\cost_\lambda(S^1, \sigma^1)] \geq \frac{\Phi_\emptyset}{3} + \sum_{i \in F\setminus S}\frac{2\Phi_i}{3|F \setminus S|}$.  Therefore,
		 \begin{align*}
		 		\sum_{\mathsf{op} \in \calQ}\nabla_{\mathsf{op}} &\leq |\calQ_\emptyset|\Phi_\emptyset + \sum_{i \in F \setminus S}|\calQ_i|\Phi_i \leq  |F|\Phi_\emptyset + 2\sum_{i \in F \setminus S}\Phi_i \\
		 		&\leq 3 |F|(\cost_\lambda(S^0, \sigma^0) - \E[\cost_\lambda(S^1, \sigma^1)]),
		 \end{align*}
		 since the third and fourth properties in the theorem imply $|\calQ_\emptyset| \leq  |F|$ and $|\calQ_i| \leq 2$ for every $i \in F \setminus S$.  Replacing $\calQ$ with each of $\calP_\rmC$, $\calP_\rmF$ and $\calP_\rmC \biguplus \calP_\rmF$, we obtain 
		\begin{align*}
			\cost_\lambda(S^0, \sigma^0) - \E[\cost_\lambda(S^1, \sigma^1)] \geq \frac1{3 |F|}\max\left\{
				\begin{array}{c}
					\cc(\sigma^0) - (\lambda f(S^*) + \cc(\sigma^*))\\
					\lambda f(S) - (\lambda f(S^*) + 2\cc(\sigma^*))\\
					\cost_\lambda(S^0, \sigma^0) - (2\lambda f(S^*) + 3\cc(\sigma^*))
				\end{array}
			\right\}.
		\end{align*}
		This finishes the proof of the lemma. 
\end{proof}

\fliterate*
\begin{proof}
	We break the procedure in two stages. The first stage contains $M_1 = O\left( |F|\log\frac{\Gamma}{\epsilon'}\right)$ iterations of the for-loop in $\mathsf{FL\mhyphen iterate}(M)$, where $M_1$ is sufficiently large.  Applying Lemma~\ref{lemma:sample-local-search} and using the third term in the $\max$ operator, for any execution of $\mathsf{sampled\mhyphen local\mhyphen search}$, we have	 
	\begin{align*}
		&\quad \E\big[\big(\cost_\lambda(S^1, \sigma^1) - (2\lambda f(S^*) + 3\cc(\sigma^*))\big)_+\big] \\
		&\leq \left(1- \frac{1}{3 |F|}\right)\big(\cost_\lambda(S^0, \sigma^0) - (2\lambda f(S^*) + 3\cc(\sigma^*))\big)_+,
	\end{align*}
	where $(S^0, \sigma^0)$ and $(S^1, \sigma^1)$ are as defined w.r.t the execution, and $x_+$ is defined as $\max\{x ,0\}$ for every real number $x$.  Notice that when $\cost_\lambda(S^0, \sigma^0) \leq 2\lambda f(S^*) + 3\cc(\sigma^*)$, the inequality holds trivially.  Truncating at $0$ is needed later when we apply the Markov inequality. 
	
	So, after $M_1$ iterations, we have 
	\begin{align*}
		&\quad \E\big[\big(\cost_\lambda(S, \sigma) - (2\lambda f(S^*) + 3\cc(\sigma^*))\big)_+\big] \\
		&\leq \left(1- \frac{1}{3 |F|}\right)^{M_1}\big(\cost_\lambda(S^\circ, \sigma^\circ) - (2\lambda f(S^*) + 3\cc(\sigma^*))\big)_+\leq \frac{\epsilon'}{2\Gamma}\opt.
	\end{align*}
	The second inequality holds since $\cost_\lambda(S^\circ, \sigma^\circ) \leq \lambda\cost(S^\circ, \sigma^\circ) \leq O(1)\opt$ and $M = O\left(\frac{|F|}{\epsilon'}\log \Gamma\right)$ is sufficiently large. Using Markov's inequality, with probability at least $1-\frac{1}{2\Gamma}$, we have at the end of the first stage, $$(\cost_\lambda(S, \sigma) - (2\lambda f(S^*) + 3\cc(\sigma^*)))_+ \leq \epsilon'\cdot\opt.$$ If the event happens, we say the first stage is successful. 
	
	We assume the first stage is successful and analyze the second stage. The second stage contains $\log_2(2\Gamma)$ phases, and each phase contains $\frac{48 |F|}{\epsilon'}$ iterations.  We focus on one phase in the stage.
	Assume that at the beginning of an iteration in the phase, we have 
	\begin{align*}
		\cc(\sigma) \leq \big(\lambda + \frac{\epsilon'}2\big) f(S^*) + \big(1+\frac{\epsilon'}2\big)\cc(\sigma^*) \text{ and } \lambda f(S) \leq \big(\lambda + \frac{\lambda\epsilon'}2\big) f(S^*) + \big(2+\frac{\lambda\epsilon'}2\big)\cc(\sigma^*).
	\end{align*}
	Then at the moment, we have $\cost(S, \sigma) \leq (1 + \lambda + \epsilon')f(S^*) + (1+2/\lambda  + \epsilon')\cc(\sigma^*) = (\alpha_\FL + \epsilon')\opt$ (obtained by adding the first inequality and $1/\lambda$ times the second inequality). Then we must have $\cost(S^{\mathsf{best}}, \sigma^{\mathsf{best}}) \leq (\alpha_\FL + \epsilon')\opt$ in the end of this execution of $\mathsf{FL\mhyphen iterate}$  since $(S^{\mathsf{best}}, \sigma^{\mathsf{best}})$ is the best solution according to the original (i.e, non-scaled) cost. 
	
	Thus, we say a phase in the second stage is successful if both inequalities hold at the end of some iteration in the phase; then we can pretend that the phase ends at the moment it is successful. If one of the two inequalities does not hold at the end of an iteration, then by Lemma~\ref{lemma:sample-local-search}, for the execution of $\mathsf{sampled\mhyphen local\mhyphen search}$ in the next iteration, we have $\cost_\lambda(S^0, \sigma^0) - \E[\cost_{\lambda}(S^1, \sigma^1)] \geq \frac{\epsilon'}{6 |F|}(f(S^*) + \cc(\sigma^*)) = \frac{\epsilon'}{6 |F|}\opt$. Then, by stopping times of martingales, in expectation, the phase stops in at most $\frac{24 |F|}{\epsilon'}$ iterations since at the beginning of the phase we have $\cost_\lambda(S, \sigma) \leq \max\{3+\epsilon', 2\lambda+\epsilon'\}(f(S^*) + \cc(\sigma^*)) \leq 4\cdot\opt$ and $\cost_\lambda(S, \sigma)$ is always positive.  By Markov's inequality, the probability that the phase does not stop early (i.e, is not successful) is at most $1/2$.  The probability that the second stage succeeds, i.e, at least one of its phases succeeds is at least $1-1/(2\Gamma)$.  Thus with probability at least $1-1/\Gamma$, both stages succeed and we have $\cost(S^\best, \sigma^\best) \leq (\alpha_\FL + \epsilon')\opt$. The number of iterations we need in the two stages is $O\left(\frac{ |F|}{\epsilon'}\log \Gamma\right)$.
\end{proof}

\end{document}